\documentclass[nocopyrightspace,reprint]{sigplanconf}

\usepackage{times}
\usepackage{multirow}
\usepackage{xspace}
\usepackage{latexsym,stmaryrd}
\usepackage{amsthm}
\usepackage{amsfonts}
\usepackage{amssymb}
\usepackage{amsmath}
\usepackage{verbatim}
\usepackage{epsfig}
\usepackage{graphics}
\usepackage{url}
\usepackage{wrapfig}

\usepackage{mathpartir}
\usepackage{code}
\usepackage{supertabular}
\usepackage[dvipsnames,usenames]{color}
\usepackage{soul}
\usepackage[all]{xy}

\usepackage{comment}
\excludecomment{insubonly}
\includecomment{intronly}

\newcommand{\subonly}[1]{#1}
\newcommand{\tronly}[1]{}

\newcommand{\out}[1] {}

\newcounter{codeLineCntr}

\reversemarginpar
\newif\ifnotes
\notestrue

\newcommand{\punt}[1]{}

\newcommand{\figref}[1]{Figure~\ref{fig:#1}}

\renewcommand{\eqref}[1]{Equation~(\ref{eq:#1})}

\newcommand{\proc}[1]{\ifmmode\mbox{\textsc{#1}}\else\textsc{#1}\fi}

  \newcommand{\func}[1]{\ifmmode\mathrm{#1}\else\textrm{#1}fi} %
\newcounter{remark}[section]

\newcommand{\im}[1]{\ensuremath{#1}}

\newcommand{\kw}[1]{\im{\mathtt{#1}}}

\newcommand{\ts}{\vdash}
\newcommand{\red}{\Downarrow}

\renewcommand{\hole}{\im{\Box}}

\newcommand{\set}[1]{\im{\{{#1}\}}}

\newcommand{\omitthis}[1]{}

\newenvironment{nop}{}{}

\newenvironment{sdisplaymath}{
\begin{nop}\small\begin{displaymath}}{
\end{displaymath}\end{nop}\ignorespacesafterend}

\newenvironment{smathpar}{
\begin{nop}\small\begin{mathpar}}{
\end{mathpar}\end{nop}\ignorespacesafterend}

\newenvironment{stackAux}[2]{%
\setlength{\arraycolsep}{0pt}
\begin{array}[#1]{#2}}{
\end{array}}

\newtheorem{theorem}{Theorem}[section]
\newtheorem{lemma}[theorem]{Lemma}

\newtheorem{corollary}[theorem]{Corollary}
\newtheorem{definition}[theorem]{Definition}
\newtheorem{proposition}[theorem]{Proposition}

\newcommand{\bnfalt}{{\bf \,\,\mid\,\,}}
\newcommand{\bnfdef}{{\bf ::=}}

\newcommand{\tyint}{\im{\kw{int}}}
\newcommand{\tybool}{\im{\kw{bool}}}
\newcommand{\tyrec}[1]{\kwrec{#1}}
\newcommand{\tyset}[1]{\kwset{#1}}

\newcommand{\semplus}{\im{\,\hat{+}\,}}

\newcommand{\semu}{\uplus}

\newcommand{\bindl}[2]{\im{{#1}.{#2}}}
\newcommand{\lbl}{\im{\ell}\xspace}

\newcommand{\kwtrue}{\im{\kw{true}}}
\newcommand{\kwfalse}{\im{\kw{false}}}
\newcommand{\kwc}{\im{\kw{c}}}
\newcommand{\kwf}{\im{\kw{f}}}

\newcommand{\kwrec}[1]{\langle #1 \rangle}
\newcommand{\kwfield}[2]{#1.#2}
\newcommand{\kwpair}[2]{\im{({#1},{#2})}}

\newcommand{\kwif}[3]{\im{\kw{if}(#1,#2,#3)}}
\newcommand{\kwsete}{\im{\emptyset}}
\newcommand{\kwset}[1]{\im{\set{#1}}}
\newcommand{\kwsets}[1]{\im{\set{#1}}}
\newcommand{\kwsetu}[2]{\im{{#1} \cup {#2}}}
\newcommand{\kwsetise}[1]{\im{\kw{empty}~{#1}}}

\newcommand{\kwlet}[3]{\im{\kw{let}~{#1}\:=\:{#2}~\kw{in}~{#3}}}

\newcommand{\kwsum}[1]{\im{\kw{sum}~{#1}}}

\newcommand{\kwcomp}[3]{\im{\bigcup\set{{#1\!}\mid{#2 \in #3}}}}

\newcommand{\trsub}{\sqsubseteq}

\newcommand{\trbslice}{\searrow}

\newcommand{\tri}{\triangleright}

\newcommand{\tremp}{\im{\hole}}
\newcommand{\mktrace}[1]{\im{{#1}}}
\newcommand{\trc}{\kwc}
\newcommand{\trf}{\kwf}

\newcommand{\trpair}[2]{\mktrace{\kwpair{#1}{#2}}}
\newcommand{\trsetise}[1]{\kw{empty}~{#1}}
\newcommand{\trsets}[1]{\{#1\}}
\newcommand{\trsetu}[2]{\mktrace{{#1}~\cup~{#2}}}

\newcommand{\trcomp}[4]{\kwcomp{#1}{#2}{#3} \tri #4}

\newcommand{\trrec}[1]{\kwrec{#1}}
\newcommand{\trfield}[2]{\kwfield{#1}{#2}}

\newcommand{\trif}[5]{\mktrace{{\kw{if}(#1,#2,#3)}}\tri_{#4} #5}
\newcommand{\trifthen}[4]{\trif{#1}{#2 }{#3}{\kw{true}}{#4}}
\newcommand{\trifelse}[4]{\trif{#1}{#2}{#3}{\kw{false}}{#4}}
\newcommand{\trlet}[3]{\mktrace{{\kw{let}}~{#1} = {#2}~\kw{in}~{#3}}}
\newcommand{\trsetsum}[1]{\mktrace{\kwsum{#1}}}

\newcommand{\trrun}{\im{\curvearrowright}}

\newcommand{\simat}[1]{\eqat{#1}}
\newcommand{\eqat}[1]{\im{\eqsim_{#1}}}

\newcommand{\dom}{\mathrm{dom}}

\newcommand{\sqleq}{\sqsubseteq}
\newcommand{\sqgeq}{\sqsupseteq}

\newenvironment{umutmathfig}{\begin{fixedCodeFrame}\begin{displaymath}}{\end{fixedCodeFrame}\end{displaymath}\nocaptionrule}

\newenvironment{amalmathfig}{\begin{sdisplaymath}}{\end{sdisplaymath}}
\newenvironment{amalsyntaxfig}{\begin{amalmathfig}\begin{array}{@{}l@{\quad}r@{~~}c@{\quad}l}}{\end{array}\end{amalmathfig}}

\newcommand{\vhole}{\Box}
\newcommand{\vany}{\Diamond}
\newcommand{\vrec}[1]{\kwrec{#1}}
\newcommand{\vset}[1]{\set{#1}}
\newcommand{\vc}{\kwc}
\newcommand{\vpair}[2]{(#1,#2)}

\newcommand{\uneval}{\mathrel{{\searrow}\!\!\!\!\!\!\!{\searrow}}}
\newcommand{\udot}{\mathbin{\dot{\cup}}}
\newcommand{\rpat}{\mathbf{rp}}
\newcommand{\spat}{\mathbf{sp}}

\newcommand{\Slicer}{\textsf{Slicer}\xspace}
\newcommand{\NRCSlicer}{\textsf{NRCSlicer}\xspace}

\begin{document}
\title{Database Queries that Explain their Work}

\authorinfo{James Cheney}{University of Edinburgh}{jcheney@inf.ed.ac.uk}
\authorinfo{Amal Ahmed}{Northeastern University}{amal@ccs.neu.edu}
\authorinfo{Umut A. Acar}{Carnegie Mellon University \& INRIA-Rocquencourt}{umut@cs.cmu.edu}

% \author{James Cheney\\
% University of Edinburgh\\
% \texttt{jcheney@inf.ed.ac.uk}
%  \and
%  Amal Ahmed\\
% Northeastern University\\
% \texttt{amal@ccs.neu.edu}
% \and
% Umut A. Acar\\
% Carnegie Mellon University \\
% \& INRIA-Rocquencourt\\
% \texttt{umut@cs.cmu.edu}}
% \exclusivelicense
% \conferenceinfo{PPDP~'14}{September 08--10, 2014, Canterbury, United Kingdom}
% \copyrightyear{2014}
% \copyrightdata{978-1-4503-2947-7/14/09}
% \doi{2643135.2643143}
\maketitle

\begin{abstract}
  Provenance for database queries or scientific workflows is often
  motivated as providing \emph{explanation}, increasing understanding
  of the underlying data sources and processes used to compute the
  query, and \emph{reproducibility}, the capability to recompute the
  results on different inputs, possibly specialized to a part
  of the output.  Many provenance systems claim to provide such
  capabilities; however, most lack formal definitions or guarantees of
  these properties, while others provide formal guarantees only for
  relatively limited classes of changes.  Building on recent work on
  provenance traces and slicing for functional programming languages,
  we introduce a detailed tracing model of provenance for multiset-valued Nested
  Relational Calculus, define trace slicing
  algorithms that extract subtraces needed to explain or recompute
  specific parts of the output, and define query slicing and
  differencing techniques that support explanation.  We state and
  prove correctness properties for these techniques and present a
  proof-of-concept implementation in Haskell.
\end{abstract}

% LocalWords:  hoc subtraces internet  workflows reproducibility formalisms
% LocalWords:  multisets differencing


%\category{D.3.0}
%{Programming Languages}
%{General}
%\category{D.3.3}
%{Programming Languages}
%{Language Constructs and Features}

%\terms Provenance, semantics. 

%\keywords  

%\input{abstract}
\section{Introduction}

Over the past decade, the use of complex computer systems in science
has increased dramatically: databases, scientific workflow systems,
clusters, and cloud computing based on frameworks such as
MapReduce~\cite{mapreduce} or PigLatin~\cite{piglatin} are now
routinely used for scientific data analysis.  With this shift to
computational science based on (often) unreliable components and noisy
data comes decreased transparency, and an increased need to understand
the results of complex computations by auditing the underlying
processes.

This need has motivated work on \emph{provenance} in databases,
scientific workflow systems, and many other
settings~\cite{buneman01icdt,moreau10ftws}.  There is now a great deal
of research on extending such systems with rich provenance-tracking
features.  Generally, these systems aim to provide
high-level explanations intended to aid the user in understanding how
a computation was performed, by recording and presenting additional
``trace'' information.

Over time, two distinct
approaches to provenance have emerged: (1) the use of \emph{annotations} propagated
through database queries to illustrate  \emph{where-provenance}
linking results to source data~\cite{buneman01icdt}, \emph{lineage} or
\emph{why-provenance} linking result records to sets of witnessing
input records~\cite{DBLP:journals/tods/CuiWW00}, or
\emph{how-provenance} describing how results were produced via
algebraic expressions~\cite{DBLP:conf/pods/2007/GreenKT07}, 
and (2) the use of graphical \emph{provenance traces} to illustrate
how workflow computations construct final results from inputs and
configuration
parameters~\cite{bose05cs,DBLP:conf/dils/HiddersKSTB07,DBLP:journals/sigmod/SimmhanPG05}.
However, to date few systems formally specify the semantics of
provenance or give formal guarantees characterizing how provenance
``explains'' results.

For example, scientists often conduct parameter sweeps to search for
interesting results. The provenance trace of such a computation may be
large and difficult to navigate.  Once the most promising
results have been identified, a scientist may want to extract just
that information that is needed to explain the result, without showing
all of the intermediate search steps or uninteresting results.
Conversely, if the results are counterintuitive, the scientist may
want to identify the underlying data that contributed to the anomalous
result.  Missier et al.~\cite{missier11ijdc} introduced the idea of a
``golden trail'', or a subset of the provenance trace that explains,
or allows reproduction of, a high-value part of the output.  They
proposed techniques for extracting ``golden trails'' using recursive
Datalog queries over provenance graphs; however, they did not propose
definitions of correctness or reproducibility.

It is a natural question to ask how we know when a proposed solution,
such as Missier et al.'s ``golden trail'' queries, correctly
explains or can be used to correctly reproduce the behavior of the
original computation.  It seems to have been taken for granted that simple
graph traversals suffice to at least overapproximate the desired
subset of the graph.  As far as we know, it is still an open question
how to define and prove such correctness properties for most
provenance techniques.  In fact, these properties might be defined and
formalized in a number of ways, reflecting different modeling choices
or requirements.  In any case, in the absence of clear statements and
proofs of correctness, claims that different forms of provenance
``explain'' or allow ``reproducibility'' are difficult to evaluate objectively.

The main contribution of this paper is to formalize and prove the
correctness of an approach to fine-grained provenance for database
queries.  We build on our approach developed in prior work, which we
briefly recapitulate.  Our approach is based on analogies between the
goals of provenance tracking for databases and workflows, and those of
classical techniques for program comprehension and analysis,
particularly \emph{program slicing}~\cite{weiser81icse} and
\emph{information flow}~\cite{sabelfeld03sac}.
% The idea of program slicing, originating in
% the work of Weiser~\cite{weiser81icse}, is to aid programmers in the
% task of debugging by erasing parts of the program that are
% \emph{irrelevant} to a property of the result (according to some
% semantically-based notion of relevance).
% Similarly, a large body of research on information flow in computer
% security~\cite{sabelfeld03sac} has considered the problem of tracking, statically or
% dynamically, how outputs depend on inputs, in order to determine
% whether an attacker can learn confidential information by observing
% the behavior of the system even if the secret values are redacted.  
Both program slicing and information flow rely critically on notions
of \emph{dependence}, such as the familiar control-flow and data-flow
dependences in programming languages.

We previously introduced a provenance model for NRC (including
difference and aggregation operations) called \emph{dependency
  provenance}~\cite{cheney11mscs}, and showed how it can be used to
compute \emph{data slices}, that is, subsets of the input to the query
that include all of the information relevant to a selected part of the
output.  Some other forms of provenance for database query languages,
such as how-provenance~\cite{DBLP:conf/pods/2007/GreenKT07}, satisfy similar formal guarantees that can be
used to predict how the output would change under certain classes of
input changes, specifically those expressible by semiring
homomorphisms.  For example, Amsterdamer et al.'s
system~\cite{amsterdamer11pvldb} is based on the semiring provenance
model, so the effects of deletions on parts of the output can be
predicted by inspecting their
provenance, but other kinds of changes are not supported.

More recently, we proposed an approach to provenance called
\emph{self-explaining computation}~\cite{cheney13pbf} and explored it
in the context of a general-purpose functional programming
language~\cite{acar13jcs,perera12icfp}.  In this approach, detailed
execution traces are used as a form of provenance. Traces explain
results in the sense that they can be \emph{replayed} to recompute the
results, and they can be \emph{sliced} to obtain smaller traces that
provide more concise explanations of parts of the output.  Trace slicing
also produces a slice of the input showing what was needed by the
trace to compute the output.  Moreover,
other forms of provenance can be extracted from traces (or slices),
and we also showed that traces can be used to compute program slices
efficiently through lazy evaluation.
%  although it is potentially expensive to do this by naively building
% the full execution trace in-memory, we showed that this cost can be
% ameliorated using a lazy evaluation strategy that does not explicitly
% construct the trace until it is needed during slicing
Finally, we
showed how traces support \emph{differential slicing} techniques that
can highlight the differences between program runs in order to explain
and precisely localize bugs in the program or errors in the input
data.

Our long-term vision is to develop self-explaining computation
techniques covering all components used in day-to-day scientific
practice.  Databases are probably the single most important such
component.  Since our previous work already applies to a
general-purpose programming language, one way to proceed would be to
simply implement an interpreter for NRC in this language, and inherit
the slicing behavior from that.  However, without some further
inlining or optimization, this naive strategy would yield traces that
record both the behavior of the NRC query and its interpreter, along
with internal data structures and representation choices whose details
are (intuitively) irrelevant to understanding the high-level behavior
of the query.

\subsection{Technical overview}

In this paper, we develop a tracing semantics and trace slicing
techniques tailored to NRC (over a multiset semantics).  This
semantics evaluates a query $Q$ over an input database
(i.e. environment $\gamma$ mapping relation names to table values),
yielding the usual result value $v$  as well as a \emph{trace} $T$.
Traces are typically large and difficult to decipher, so we consider a
scenario where a user has run $Q$, inspected the results $v$, and requests
an explanation for a part of the result, such as a field value of a
single record.  As in our previous work for functional programs, we
use partial values with ``holes'' $\vhole$ to describe parts of the
output that are to be explained.  For example, if the result of a
program is just a pair $(1,2)$ then the pattern $(1,\vhole)$ can be
used to request an explanation for just the first component.  Given a
partial value $p$ matching the output, our approach computes a
``slice'' consisting of a partial trace and
a partial input environment, where components not necessary for
recomputing the explained output part $p$ have been deleted.

The main technical contribution of this paper over our previous
work~\cite{acar13jcs,perera12icfp} is its treatment of tracing and
slicing for collections.  There are two underlying technical
challenges; we illustrate both (and our solutions) via a simple
example query $Q = \sigma_{A < B}(R) \cup \rho_{A \mapsto B,B \mapsto
  A}(\sigma_{A\geq B}(R))$ over a table $R$ with attributes $A,B$.
Here, $\sigma_\phi$ is relational selection of all tuples satisfying a
predicate $\phi$ and $\rho_{A \mapsto B,B \mapsto A}$ is renaming.
Thus, $Q$ simply swaps the fields of records where $A \geq B$, and
leaves other records alone.

The first challenge is how to address elements of multisets reliably
across different executions and support propagation of addresses in
the output backwards towards the input.  Our solution is to use a
mildly enriched semantics in which multiset elements carry explicit
labels; that is, we view multisets of elements from $X$ as functions
$I \to X$ from some index set to $X$.  
For example, if $R$ is labeled as follows and we use
this enriched semantics to evaluate the above query $Q$ on $R$, we get
a result:
\[R = \begin{array}{c|ccc}
  id & A & B & C\\
\hline
 {}[r_1] & 1 & 2 & 7\\
 {}[r_2] & 2 & 3 & 8\\
 {}[r_3] & 4 & 3 & 9
\end{array}\quad
Q(R) = \begin{array}{c|ccc}
  id & A & B & C\\
\hline
{}[1,r_1] & 1 & 2 & 7\\
{}[1,r_2] & 2 & 3 & 8\\
{}[2,r_3] & 3 & 4 & 9
\end{array}
\]
where in each case the $id$ column contains a distinct index $r_i$.
In $Q(R)$, the first `1' or `2' in each index indicates whether the
row was generated by the left or right subexpression in the union
('$\cup$').  % This idea is reminiscent of the deterministic tree data
% model used, for example, in the original work on where- and
% why-provenance~\cite{buneman01icdt}; however, our application to trace
% and query slicing appears to be new.
In general, we use sequences
of natural numbers $[i_1,\ldots,i_n] \in \mathbb{N}^*$ as indices, and
we maintain a stronger invariant: the set of indexes used in a
multiset must form a \emph{prefix code}.  We define a semantics for
NRC expressions over such collections that is fully deterministic and
does not resort to generation of fresh intermediate labels; in
particular, the union operation adjusts the labels to maintain
distinctness.  

The second technical challenge involves extending \emph{patterns} for
partial collections.  In our previous work, any subexpression of a
value can be replaced by a hole.
% In our previous work, we defined slicing using
% partial expressions where some subexpressions may be omitted, or
% replaced with \emph{holes} $\vhole$. For example, a pattern $p_1 =
% (\vhole,1)$ indicates that we are interested in how the second
% component of a pair was computed, but not the first component.  The
% other component is replaced by a \emph{hole} $\vhole$ that,
% intuitively, stands for an unknown, irrelevant value. We can then
% define slicing as an algorithm that takes a trace and an output
% pattern and returns a sliced trace and subpattern of the input,
% discarding information that is not needed to recompute the value
% matching the output pattern.  This approach, in which any
% subexpression of a value can be replaced by a hole, 
This works well in a
conventional functional language, where typical values (such as lists
and trees) are essentially initial algebras built up by structural
induction.  However, when we consider unordered collections such as bags, our previous
approach becomes awkward.  

For example, if we want to use a pattern to
focus on only the $B$ field of the second record
in query result $Q(R)$, we can only do this by deleting the other
record values, and the smallest such pattern is
$\{[1,r_1].\vhole,[2,r_2].\kwrec{A{:}\vhole,B{:}3,C{:}\vhole},[2,r_3].\vhole\}$.
This is tolerable if there are only a few elements, but if there are hundreds or
millions of elements and we are only interested in one, this is a
significant overhead.  Therefore, we introduce \emph{enriched
  patterns} that allow us to replace entire subsets with holes, for
example, $p' = \{[2,r_2]. \kwrec{B{:}3;\vhole}\} \udot \vhole$.  Enriched patterns
are not just a convenience; we show experimentally that they allow
traces and slices to be smaller (and computed faster) by an order of
magnitude or more.

\subsection{Outline}

The rest of this paper is structured as follows. Section~\ref{sec:nrc}
reviews the (multiset-valued) Nested Relational Calculus and presents
our tracing semantics and deterministic labeling scheme.
Section~\ref{sec:slicing} presents trace slicing, including a simple
form of patterns.  Section~\ref{sec:patterns} shows how to enrich our
pattern language to allow for partial record and partial set patterns,
which considerably increase the expressiveness of the pattern
language, leading to smaller slices.  Section~\ref{sec:qdslicing}
presents the query slicing algorithm and shows how to compute
differential slices.  Section~\ref{sec:examples} presents additional
examples and discussion.  Section~\ref{sec:implementation} presents
our implementation demonstrating the
benefits of laziness and enriched patterns for trace slicing.
Section~\ref{sec:related} discusses related and future work and
Section~\ref{sec:concl} concludes.

Due to space limitations, and in order to make room for examples and
high-level discussion, some (mostly routine) formal details and proofs
are relegated to 
\begin{insubonly}
  a companion technical report~\cite{tr}.
\end{insubonly}
\begin{intronly}
  the appendix.
\end{intronly}

\section{Traced Evaluation for NRC}
\label{sec:nrc}

The nested relational calculus~\cite{buneman95tcs} is a simply-typed
core language with collection types that can express queries on nested
data similar to those of SQL on flat relations, but has simpler syntax and cleaner
semantics.  In this section, we show how to extend the ideas and
machinery developed in our previous work on traces and slicing for
functional languages~\cite{perera12icfp} to NRC.  Developing formal
foundations for tracking provenance in the presence of unordered
collections presents a number of challenges not encountered in the
functional programming setting, as we will explain.

% Specifically, the language we consider in this section adds
% collections to the $\lambda$-calculus of the previous section.
% However, the results we present hold regardless of whether the
% language has higher-order functions.  

\begin{figure}
\[\small
\begin{array}{@{}l@{~~}r@{~}c@{~~}l}
\mbox{Operations} & \kwf & \bnfdef & 
+ \bnfalt 
- \bnfalt 
* \bnfalt 
/ \bnfalt 
= \bnfalt
< \bnfalt
\leq \bnfalt \cdots
\\[1mm]
 \mbox{Expressions} & 
e & \bnfdef & 
\kwc \bnfalt
\kwf(e_1,\ldots,e_n) \bnfalt
x \bnfalt
\kwlet{e}{x}{e'}
\\[1mm]
&&\bnfalt &
\kwrec{A_1 : e_1,\ldots,A_n:e_n} \bnfalt
\kwfield{e}{A}
%\\[1mm]
%&&\bnfalt&
\bnfalt \kwif{e}{e'}{e''}
\\[1mm]
&&\bnfalt&
\kwsete \bnfalt
\kwsets{e} \bnfalt 
\kwsetu{e_1}{e_2} \bnfalt
\kwcomp{e'}{x}{e} 
\\[1mm]
&&\bnfalt&
\kwsetise{e} \bnfalt
\kwsum{e} \bnfalt \cdots
%%\\
% %&& &
\\[1mm]
\mbox{Labels} &
\lbl & \bnfdef &  \lbl.\lbl' \mid \epsilon \mid n
\\[1mm]
\mbox{Values} &
v & \bnfdef &  
\kwc \bnfalt \kwrec{A_1:v_1,\ldots A_n:v_n} \\[1mm]
&&\bnfalt&
\kwsete \bnfalt 
\kwset{\bindl{\lbl_1}{v_1}, \ldots, \bindl{\lbl_n}{v_n}}
\\[1mm]
\mbox{Environments}
& \gamma & \bnfdef & 
[x_1 \mapsto v_1, \ldots, x_n \mapsto v_n]
\\[1mm]

\mbox{Traces} & 
T & \bnfdef & \cdots \bnfalt
% \trc \bnfalt 
% \trf(T_1,\ldots,T_n) \bnfalt
% x \bnfalt
% \trlet{x}{T_1}{T_2}
% \\[1mm]
% &&\bnfalt &
% \trrec{A_1 : T_1,\ldots,A_n:T_n} \bnfalt
% \trfield{T}{A}
% \\[1mm]
% &&\bnfalt&
\trifthen{T}{e'}{e''}{T} \bnfalt
\trifelse{T}{e'}{e''}{T}
\\[1mm]
&&\bnfalt&
% \emptyset \bnfalt
% \trsets{T} \bnfalt 
% \trsetu{T_1}{T_2} \bnfalt
%\trflatten{T} \bnfalt
%\\
%& & \bnfalt & 
\trcomp{e}{x}{T}{\Theta} 
% \\[1mm]
% & &  \bnfalt& 
% \trsetise{T}  \bnfalt
% \trsetsum{T} \bnfalt
% \cdots
\\[1mm]
\mbox{Trace Sets} & 
\Theta & \bnfdef & \{\lbl_1.T_1,\ldots,\lbl_n.T_n\}\\[1mm]
\mbox{Types} & 
\tau & \bnfdef & 
\tyint \mid \tybool \mid \tyrec{A_1:\tau_1,\ldots,A_n:\tau_n} \mid \tyset{\tau}\\[1mm]
\mbox{Type Contexts} &
\Gamma & \bnfdef &
x_1 : \tau_1,\ldots,x_n:\tau_n
\end{array}
\]
%\nocaptionrule
\caption{NRC expressions, values, traces, and types.}
\label{fig:nrc-syntax}
\end{figure}

\subsection{Syntax and Dynamic Semantics}
\label{sec:nrc::syn}
\begin{figure*}[tb]
\fbox{$\gamma,e \red v,T$}
\vspace{-2ex}
\begin{smathpar}
\inferrule*{\strut
}{
\gamma, \kwc \red \kwc,\trc
}
\and
\inferrule*{\gamma,e_1 \red \kwc_1,T_1\\
\cdots\\
\gamma,e_n \red \kwc_n,T_n
}{
\gamma,\kwf(e_1,\ldots,e_n) \red \hat{\kwf}(\kwc_1,\ldots,\kwc_n),\trf(T_1,\ldots,T_n)
}
\and
\inferrule*{\strut
}{
\gamma,x \red \gamma(x),x
}
\and
\inferrule*{
\gamma,e_1 \red v_1,T_1\\
\gamma[x\mapsto v_1],e_2 \red v_2,T_2
}{
\gamma,\kwlet{x}{e_1}{e_2} \red v_2,\trlet{x}{T_1}{T_2}
}
\and
\inferrule*{
\gamma,e_1 \red v_1,T_1\\
\cdots\\
\gamma,e_n \red v_n,T_n
}{
\gamma,\kwrec{A_1{:}e_1,\ldots,A_n{:}e_n} \red \kwrec{A_1{:}v_1,\ldots,A_n{:}v_n},\trrec{A_1{:}T_1,\ldots,A_n{:}T_n}
}
\and
\inferrule*{
\gamma,e \red \kwrec{A_1{:}v_1,\ldots,A_n{:}v_n},T
}{
\gamma,\kwfield{e}{A_i} \red v_i, \trfield{T}{A_i}
}
\and
\inferrule*{
\gamma,e \red \kwtrue,T \\
\gamma,e_1 \red v_1,T_1
}{
\gamma,\kwif{e}{e_1}{e_2} \red v_1, \trifthen{T}{e_1}{e_2}{T_1}
}
\and
\inferrule*{
\gamma,e \red \kwfalse,T \\
\gamma,e_2\red v_2,T_2
}{
\gamma,\kwif{e}{e_1}{e_2} \red v_2, \trifelse{T}{e_1}{e_2}{T_2}
}
\and
\inferrule*
{
\strut
}
{
\gamma,\emptyset \red \emptyset, \emptyset
}
\and
\inferrule*
{
\gamma, e \red v, T 
}
{\gamma,\kwsets{e} \red \kwsets{\epsilon.v}, \trsets{T}}
\and
\inferrule*
{
\gamma,e_1 \red v_1, T_1 \\
\gamma,e_2 \red v_2, T_2
}
{
\gamma,\kwsetu{e_1}{e_2} 
\red 1 \cdot v_1 \semu 2\cdot v_2, \trsetu{T_1}{T_2}
}
\and
%
% \inferrule*
% {
% \gamma,e \red \kwset{\bindl{\lbl_1}{v_1}, \ldots, \bindl{\lbl_n}{v_n}}, T 
% }
% {
% \gamma,\kwflatten{e} \red v_1 \semu \ldots \semu v_n,\trflatten{T}
% }
%
%\and
%
\inferrule*{
\gamma,e \red v,T \\
\gamma,x\in v, e' \red^* v',\Theta
}{
\gamma,\kwcomp{e'}{x}{e} \red v', \trcomp{e'}{x}{T}{\Theta}
}
\and
\inferrule*
{
\gamma,e \red \kwset{\bindl{\lbl_1}{v_1}, \ldots, \bindl{\lbl_n}{v_n}}, T 
%\\
%v = v_1 \semplus \ldots \semplus v_n
}
{
\gamma,\kwsum{e} \red v_1 \hat{+} \ldots \hat{+} v_n, \trsetsum{T}
}
\\
\inferrule*
{
\gamma,e \red v, T \\
v  = \kwsete
}
{
\gamma,
\kwsetise{e} \red \kwtrue, \trsetise{T}
}
\and
\inferrule*
{
\gamma,e \red v, T \\
v  \not= \kwsete
}
{
\gamma,\kwsetise{e} \red \kwfalse, \trsetise{T}
}
%
% \and
% %
% \inferrule*
% {
% \gamma,e \red \set{\bindl{\lbl_1}{v_1}, \ldots, \bindl{\lbl_n}{v_n}}, T
% \\
% \gamma[x\mapsto v_i], {e'} \red v_i', T_i' \mbox{~~~~(for~$i \in \set{1,\ldots,n}$)}
% }
% {
% \hspace{-10em}
% \gamma,\kwcomp{e'}{x}{e}
% \red
% \lbl_1\cdot v_1' \semu \cdots \semu \lbl_n \cdot v_n',
% \\\\
% \qquad\qquad\qquad\qquad
% \trcomp{e'}{x}{T}{\{\lbl_1.T_1,\ldots,\lbl_n.T_n\}}
% } 
\end{smathpar}
\fbox{$\gamma,x\in v, e \red^* v,\Theta$}
\vspace{-2ex}
\begin{smathpar}
\inferrule*{\strut}
{
\gamma,x\in \emptyset, e \red^* \emptyset, \emptyset
}
\and
\inferrule*{
\gamma,x \in v_1, e \red^* v_1',\Theta_1\\
\gamma,x \in v_2, e \red^* v_2', \Theta_2
}{
\gamma,x\in v_1 \semu v_2, e \red^* v_1' \semu v_2', \Theta_1 \semu
\Theta_2
}
\and 
\inferrule*{
\gamma[x\mapsto v], e \red v', T
}{
\gamma,x\in \{\lbl.v\}, e \red^* \lbl \cdot v', \{\lbl.T\}
}
\end{smathpar}
%\nocaptionrule
\caption{Traced evaluation.}
\label{fig:nrc-dynamic}
\end{figure*}

\figref{nrc-syntax} presents the abstract syntax of NRC expressions,
values, and traces.  The expression $\kwsete$ denotes the empty
collection, $\kwsets{e}$ constructs a singleton collection, and
$\kwsetu{e_1}{e_2}$ takes the (multiset) union of two collections.
The operation $\kwsum{e}$ computes the sum of a collection of
integers, while the predicate $\kwsetise{e}$ tests whether the
collection denoted by $e$ is empty.  Additional aggregation operations
such as count, maximum and average can easily be accommodated.
Finally, the comprehension operation $\kwcomp{e'}{x}{e}$ iterates over
the collection obtained by evaluating $e$, evaluating $e'(x)$ with $x$
bound to each element of the collection in turn, and returning a
collection containing the union of all of the results.  We sometimes
consider pairs $(e_1,e_2)$, a special case of records
$\vrec{\#_1:e_1,\#_2:e_2}$ using two designated field names $\#_1$ and
$\#_2$.  Many trace forms are similar to those for expressions; only
the differences are shown.

Labels are sequences $ \lbl = [i_i,\ldots,i_n] \in \mathbb{N}^*$,
possibly empty.  The empty sequence is written $\epsilon$, and labels
can be concatenated $\lbl \cdot \lbl'$; concatenation is associative.
Record field names are written $A,B,A_1,A_2,\ldots$.

Values in NRC include constants $\vc$, which we assume include at
least booleans and integers.  Record values are essentially partial
functions from field names to values, written
$\kwrec{A_1:v_1,\ldots,A_n:v_n}$. Collection values are essentially
partial, finite-domain functions from labels in $\mathbb{N}^*$ to
values, which we write $\{\lbl_1.v_1,\ldots,\lbl_n.v_n\}$.  Since they
denote functions, collections and records are identified up to
reordering of their elements, and their field names or labels are
always distinct.  We write $\lbl \cdot v$ for the operation
that prepends $\lbl$ to each of the labels in a set $v$, that is,
\[\lbl \cdot \{\lbl_1.v_1,\ldots,\lbl_n.v_n\} = \{\lbl \cdot
\lbl_1.v_1,\ldots,\lbl \cdot \lbl_n.v_n\}\;.\]
Other operations on labels and labeled collections will be introduced in due course.

The labels on the elements of a collection provide us with a
persistent address for a particular element of the collection.  This
capability is essential when asking and answering provenance queries
about parts of the source or output data, and when tracking
fine-grained dependencies.% Developing elegant
% formalisms for pointing at and tracking elements of collections has been
% surprisingly difficult for both us and
% others~\cite{buneman01icdt,buneman08tods,DBLP:conf/dils/HiddersKSTB07}:
% the seemingly benign choice of whether to use labels or paths or
% colors can have tremendous impact on the simplicity of defining
% provenance-tracking algorithms, as well as stating and proving
% correctness criteria for them. Thus, we consider the formal details
% (explained below) of how we track elements of collections during evaluation,
% replay, and slicing to be an important contribution of our work.

Both expressions and traces are subject to a type system.  NRC types
include collection types $\set{\tau}$ which are often taken to be
sets, bags (multisets), or lists, though in this paper, we consider
multiset collections only.  However, types do not play a significant
role in this paper so the typing rules are omitted.  For expressions,
the typing judgment $\Gamma \vdash e : \tau$ is standard and
the typing rules for trace well-formedness $\Gamma \vdash T : \tau$
are presented in
\begin{insubonly}
  the companion technical report~\cite{tr}.
\end{insubonly}
\begin{intronly}
  Appendix~\ref{app:auxiliary}.
\end{intronly}
%\figref{nrc-trace-types}.

%\input{fig-nrc-trace-types}

\paragraph*{Traced evaluation}

NRC traces include a trace form corresponding to each of the
expressions described above.  The structure of the traces is best
understood by
inspecting 
\begin{intronly}
  the typing rules (\figref{nrc-trace-types}) and 
\end{intronly}
the traced evaluation rules (\figref{nrc-dynamic}), which define a
judgment $\gamma,e \red v,T$ indicating that evaluating an expression
$e$ in environment $\gamma$ yields a value $v$ and a trace $T$.  We
assume an environment $\Sigma$ associating constants and function
symbols with their types, and write $\hat{\kwf}$ or $\hat{+}$ for the
semantic operations corresponding to $f$ or $+$, and so on. In most
cases, the trace form is similar to the expression form; for example
the trace of a constant or variable is a constant trace $\trc$, the
trace of a primitive operation $\kwf(e_1,\ldots,e_n)$ is a primitive
operation trace $\trf(T_1,\ldots,T_n)$ applied to the traces $T_i$ of
the arguments $e_i$, the trace of a record expression is a trace
record constructor $\trrec{A_1:T_1,\ldots,A_n:T_n}$, and the trace of
a field projection $\kwfield{e}{A}$ is a trace $\trfield{T}{A}$.
Also, the trace of a let-binding is a let-binding trace
$\trlet{x}{T_1}{T_2}$, where $x$ is bound in $T_2$.  In these cases,
the traces mimic the expression structure.

The traced evaluation rules for conditionals illustrate that traces
differ from expressions in recording control flow decisions.  The
trace of a conditional is a conditional trace
$\trif{T}{e_1}{e_2}{b}{T'}$ where $T$ is the trace of the conditional
test, $b$ is the Boolean value of the test $e_1$,
and $T'$ is the trace of the taken branch.  The expressions $e_1$ and
$e_2$ are not strictly necessary but retained to preserve structural
similarity to the original expression.
 
  The trace of $\emptyset$ is a constant trace $\emptyset$.  To
  evaluate a singleton-collection constructor $\set{e}$, we evaluate
  $e$ to obtain a value $v$ and return the singleton
  $\set{\epsilon.v}$ with empty label $\epsilon$.  We return the
  singleton trace $\set{{T}}$ recording the trace for the evaluation
  of the element.  To evaluate the union of two expressions, we
  evaluate each one and take the semantic union (written $\semu$)
  of the resulting collections, with a `1' or `2' concatenated onto the
  beginning of each label to reflect whether each element came from
  the first or second part of the union; the union trace $T_1 \cup
  T_2$ records the traces for the evaluation of the two
  subexpressions.  For $\kwsum{e}$, evaluating $e$ yields a collection of
  numbers whose sum we return, together with a sum trace
  $\trsetsum{T}$ recording the trace for evaluation of $e$.
  Evaluation of emptiness tests $\kwsetise{e}$ is analogous, yielding a trace
  $\trsetise{T}$.

  To evaluate a comprehension $\kwcomp{e'}{x}{e}$, we first evaluate
  $e$, which yields a collection $v$ and trace $T$, and then (using
  auxiliary judgment $\gamma,x\in v, e' \red^* v', \Theta$) evaluate
  $e'$ repeatedly with $x$ bound to each element $v_i$ of the
  collection $v$ to get resulting values $v_i'$ and corresponding
  traces $T_i'$.  We return a new collection $v' = \{\lbl_1 \cdot
  v_1', \ldots,\lbl_n \cdot v_n'\}$; similarly we return a labeled set
  of traces $\Theta = \{\lbl_1.T_1,\ldots,\lbl_n.T_n\}$.  (Analogously
  to values, trace sets are essentially finite partial functions from
  labels to traces).  For each of these collections, we prepend the
  appropriate label $\lbl_i$ of the corresponding input element.

  A technical point of note is that the resulting trace $T_i$ may
  contain free occurrences of $x$.  As in our trace semantics for
  functional programs, these variables serve as markers in $T_i$ that
  will be critical for the trace replay semantics.  The comprehension
  trace records, using the notation
  $\trcomp{e'}{x}{T}{\{\lbl_1.T_1,\ldots,\lbl_n.T_n\}}$, that the
  trace $T$ was used to compute a multiset $v$, and $x$ was bound to each
  element $\lbl_i.v_i$ in $v$ in turn, with trace $T_i$ showing how
  the corresponding subset of the result was computed. The
  comprehension trace also records the expression $e'$ and bound
  variable $x$, which are again not strictly necessary but preserve
  the structural similarity to the original expression.

At this point it is useful to provide some informal motivation for the
labeling semantics, compared for example to other semantics that use
annotations or labels as a form of provenance.  We do not view the
labels themselves as
provenance; instead, they provide a useful infrastructure for traces, which do
capture a form of provenance.  Moreover, by calculating the label of
each part of an intermediate or final result deterministically (given
the labels on the input), we provide a way to reliably refer to parts
of the output, which otherwise may be unaddressable in a
multiset-valued semantics.  This is essential for supporting
compositional slicing for operations such as let-binding or
comprehension, where the output of one subexpression becomes a part of
the input for another.

A central point of our semantics is that evaluation preserves the
property that labels uniquely identify the elements of each multiset.
This naturally assumes that the labels on the input collections are
distinct.  In fact, a stronger property is required: evaluation preserves
the property that set labels form a \emph{prefix code}.  In the
following, we write $x \leq y$ to indicate that sequence $x$ is a
prefix of sequence $y$.

\begin{definition}
  A \emph{prefix code} over $\Sigma$ is a set of sequences $L
  \subseteq \Sigma^*$ such that for every $x,y \in L$, if $x \leq y$
  then $x = y$. A \emph{sub-prefix code} of a prefix code $L$ is a
  prefix code $L'$ such that for all $x \in L$ there exists $y \in L'$
  such that $y \leq x$. We write $L' \leq L$ to indicate that $L'$ is
  a sub-prefix code of $L$.  We say that $L$ and $L'$ are
  \emph{prefix-disjoint} when no element of $L$ is a prefix of an
  element of $L'$ and vice versa.
\end{definition}

Let $v$ be a collection $v = \{\lbl_1.v_1,\ldots,\lbl_n.v_n\}$.  We
define the domain of $v$ to be $\dom(v) = \{\lbl_1,\ldots,\lbl_n\}$.
.  We say that a value or value environment is \emph{prefix-labeled}
if for every collection $v$ occurring in it, the labels
$\lbl_1,\ldots,\lbl_n$ are distinct and $\dom(v)$ is a prefix
code. Similarly, we say that a trace is prefix-labeled if every
labeled trace set $\Theta = \{\lbl_1.T_1,\ldots,\lbl_n.T_n\}$ is
prefix-labeled.
\begin{theorem}\label{thm:prefix-labeled}
  If $\gamma$ is prefix-labeled and $\gamma,e \red v,T$ then $v$
  and $T$ are both prefix-labeled.  Moreover, if $\gamma$ and $v$
  are prefix-labeled and $\gamma,x\in v, e \red^* v',\Theta$ then
  $v'$ and $\Theta$ are prefix-labeled, and in addition
  $\dom(\Theta) = \dom(v) \leq \dom(v')$.
\end{theorem}
  \begin{proof}
    By induction on derivations. The key cases are those for union,
    which is straightforward, and for comprehensions.  For the latter
    case we need the second part, to show that whenever $\gamma$ and
    $v$ are prefix-labeled, if $\gamma,x\in v, e \red^* v',\Theta$
    then $v'$ and $\Theta$ are prefix-labeled, and in addition
    $\dom(\Theta) = \dom(v)$ and $\dom(v)$ is a sub-prefix code of
    $\dom(v')$.
  \end{proof}

The prefix code property is needed later in the slicing algorithms, when
we will need it to match elements of collections produced by
comprehensions with corresponding elements of the trace set $\Theta$.
From now on, we assume that all values, environments, and traces are
prefix-labeled, so any labeled set is
assumed to have the prefix code property.

\begin{figure}[tb]
\fbox{$\gamma, T \trrun v$}
\vspace{-2ex}
\begin{smathpar}
\inferrule*{\strut
}{
\gamma, \kwc \trrun \kwc
}
\and
\inferrule*{\gamma,T_1 \trrun \kwc_1\\
\cdots\\
\gamma,T_n \trrun \kwc_n
}{
\gamma,\kwf(T_1,\ldots,T_n) \trrun \hat{\kwf}(\kwc_1,\ldots,\kwc_n)
}
\and
\inferrule*{\strut
}{
\gamma,x \trrun \gamma(x)
}
\and
\inferrule*{
\gamma,T_1 \trrun v_1\\
\gamma[x\mapsto v_1],T_2 \trrun v_2
}{
\gamma,\kwlet{x}{T_1}{T_2} \trrun v_2
}
\and
\inferrule*{
\gamma,T_1 \trrun v_1\\
\cdots\\
\gamma,T_n \trrun v_n
}{
\gamma,\kwrec{A_1{:}T_1,\ldots,A_n{:}T_n} \trrun \kwrec{A_1{:}v_1,\ldots,A_n{:}v_n}
}
\and
\inferrule*{
\gamma,T\trrun \kwrec{A_1{:}v_1,\ldots,A_n{:}v_n}
}{
\gamma,\kwfield{T}{A_i} \trrun v_i
}
\and
\inferrule*%[right=RIf$_1$]
{
\gamma,T \trrun \kwtrue \\
\gamma,T_1 \trrun v_1
}{
\gamma,\trifthen{T}{e_1}{e_2}{T_1} \trrun v_1
}
%
% \and
% \inferrule*{
% \gamma,T \trrun \kwfalse \\
% \gamma,e_2\red v_2
% }{
% \gamma,\trifthen{T}{e_1}{e_2}{T_1} \trrun v_2
% }
%
% \and
% \inferrule*{
% \gamma,T \trrun \kwtrue \\
% \gamma,e_1 \red v_1
% }{
% \gamma,\trifelse{T}{e_1}{e_2}{T_2} \trrun v_1
% }
%
\and
\inferrule*%[right=RIf$_2$]
{
\gamma,T \trrun \kwfalse \\
\gamma,T_2\trrun v_2
}{
\gamma,\trifelse{T}{e_1}{e_2}{T_2} \trrun v_2
}%
\and
\inferrule*
{\strut}
{\gamma,\emptyset \trrun \emptyset}
\and
\inferrule*
{
\gamma,  T \trrun v
}
{
\gamma, \trsets{T} \trrun \kwsets{\epsilon.v}
}
\and
\inferrule*
{
\gamma, T_1 \trrun v_1
\\
\gamma, T_2 \trrun v_2
\\
}
{
\gamma,\trsetu{T_1}{T_2} \trrun 1 \cdot v_1 \semu 2 \cdot v_2
}
\and
%
% \inferrule*
% {
% \gamma, T \trrun \kwset{\bindl{\lbl_1'}{v_1},\ldots,\bindl{\lbl_m'}{v_m}} 
% %\\
% %v = v_1\semu \ldots \semu v_n
% }
% {
% \gamma, \trflatten{T} \trrun v_1\semu \ldots \semu v_m
% }
%
\and
\inferrule*
{
\gamma,T \trrun \kwset{\bindl{\lbl_1}{v_1},\ldots,\bindl{\lbl_n}{v_n}} 
%\\
%v = v_1\semplus \ldots \semplus v_n
}
{
\gamma, \trsetsum{T} \trrun v_1\semplus \ldots \semplus v_n
}
\\
\inferrule*
{
\gamma,T \trrun v \\ v = \emptyset
}
{
\gamma, \trsetise{T} \trrun \kwtrue
}
\and
\inferrule*
{
\gamma,T \trrun v \\ v \not= \emptyset
}
{
\gamma, \trsetise{T} \trrun \kwfalse
}
\and
\inferrule*%[right=RComp]
{
\gamma, T \trrun v \\
 \gamma,x\in v, \Theta \trrun^* v'
}
{
\gamma, \trcomp{e}{x}{T}{\Theta}
\trrun
v'
}
\end{smathpar}
\fbox{$\gamma,x\in v, \Theta \trrun^* v'$}
\vspace{-1ex}
\begin{smathpar}
\inferrule*%[right=REmpty$^*$]
{
\strut
}
{
 \gamma,x\in \emptyset, \Theta \trrun^* \emptyset
}
\and
\inferrule*%[right=RUnion$^*$]
{
\gamma,x\in v_1,\Theta \trrun^* v_1'\\
\gamma,x\in v_2,\Theta \trrun^* v_2'\\
}
{
\gamma,x\in v_1\semu v_2,\Theta \trrun^* v_1'\semu v_2'\\
}
\and%
\inferrule*%[right=RSng$^*$]
{
 \lbl_i \in \dom(\Theta) \\
 \gamma[x\mapsto v_i],\Theta(\lbl_i) \trrun v_i' 
}
{
 \gamma,x\in \set{\lbl_i.v_i}, \Theta \trrun^* \lbl_i \cdot v_i'
}
%
% \and
% %
% \inferrule*
% {
%  \lbl_0 \notin \dom(\Theta) \\
%  \gamma[x\mapsto v_0],e \red v_0' 
% }
% {
%   \gamma,x\in \set{\lbl_0.v_0}, \Theta,e \trrun^* \lbl_0\cdot v_0'
% }
\end{smathpar}
\caption{Trace replay.}
\label{fig:nrc-trace-run}
\end{figure}

We will use the following query as a running example.
\[Q =
\kwcomp{\kwif{x.B=3}{\{\kwrec{A{:}x.A,B{:}x.C}\}}{\{\}}}{ x}{R}\] 
This is a simple selection query; it identifies records in $R$ that
have $B$-value of 3, and returns record $\kwrec{A{:}x.A,B{:}x.C}$
containing $x$'s $A$ value and its $C$ value renamed to $B$.
The result of $Q$ on the input $R$ in the introduction is $Q(R) =
\{[r_2].\kwrec{A{:}2,B{:}8},[r_3].\kwrec{A{:}4,B{:}9}\}$, and the
trace is:
\[
\begin{array}{ll}
T =   \trcomp{\_}{x}{R}{\{&[r_1].\trifelse{x.B=3}{\_}{\_}{\{\}},\\
& [r_2].\trifthen{x.B=3}{\_}{\_}{\{\_\}},\\
&[r_3].\trifthen{x.B=3}{\_}{\_}{\{\_\}}~~ \}}\\
\end{array}
\]
where $\_$ indicates omitted (easily inferrable) subexpressions. 

\paragraph*{Replay} We introduce a judgment $\gamma,T \trrun v$ for
\emph{replaying} a trace on a (possibly different) environment
$\gamma$.  The rules for replaying NRC traces are presented in
\figref{nrc-trace-run}.  Many of the rules are straightforward or
analogous to the corresponding evaluation rules.  Here we only discuss
the replay rules for conditional and comprehension traces.

For the conditional rules, the basic idea is as follows.  If replaying
the trace $T$ of the test yields the same boolean value $b$ as
recorded in the trace, we replay the trace of the taken branch.  If
the test yields a different value, then replay fails.

To replay a comprehension trace $\trcomp{e}{x}{T}{\Theta}$, rule
$\textsc{RComp}$ first replays trace $T$ to update the set of elements
$v$ over which we will iterate.  We define a separate judgment
$\gamma,x\in v, \Theta \trrun^* v'$ to iterate over set $v$ and replay
traces on the corresponding elements. For elements $\lbl_i \in \dom(\Theta)$, we replay the
corresponding trace $\Theta(\lbl_i)$.  Replay fails if $v$ contains
any labels not present in $\Theta$.  

Replaying a trace can fail if either the branch taken in a conditional
test differs from that recorded in the trace, or the intermediate set
obtained from rerunning a comprehension trace includes values whose
labels are not present in the trace.  This means, in particular, that
changes to the input can be replayed if they only change base values
or delete set elements, but changes leading to additions of new labels
to sets involved in comprehensions typically cannot be replayed.

% Note that the first $\trrun^*$ rule, after replaying the trace $T_i'$
% to get an updated value $v_i'$, adds $\bindl{\lbl_i'}{v_i'}$ to the
% final result set (thus \emph{reusing} an old label and maximizing
% further replay opportunities).  The second rule, however, must
% generate a fresh label for the result of running $e'(x)$

Returning to the running example, suppose we change field $B$ of row $r_1$ of $R$
from 2 to 5.  This change has no effect on the control flow choice
taken in $Q$, and replaying the trace $T$ succeeds.  Likewise,
changing the $A$ or $C$ field of any column of $R$ has no effect,
since these values do not affect the control flow choices in $T$. However, changing
the $B$ field of a row to $3$ (or changing it from $3$ to something
else) means that replay will fail.

\subsection{Key properties}
\label{sec:nrc::properties}

Before moving on to consider trace slicing, 
we identify some properties that formalize the intuition that traces are
consistent with and faithfully record execution.

Evaluation and traced evaluation are deterministic:
% First, note that NRC evaluation is deterministic, up to reordering of
% labeled set values.  Similarly, traced evaluation is deterministic up to
% reordering of labeled trace sets, as is trace replay:
\begin{proposition}[Determinacy]
  If $\gamma,e \red v,T$ and $\gamma,e \red v',T'$ then $v = v'$ and
  $T = T'$. If $\gamma,T \trrun v$ and $\gamma,T \trrun
  v'$ then $v = v'$.
\end{proposition}
When the trace is irrelevant, we write $\gamma,e \red v$ to indicate
that $\gamma,e \red v,T$ for some $T$.  % If we drop the labels on
% collections, this semantics for queries coincides with the standard
% denotational semantics of NRC queries in terms of multisets; however,
% for our purposes it is more convenient to use this operational presentation.

\begin{intronly}
  Traced evaluation is type-safe and produces well-typed traces, and
  trace replay is also type-safe.
  \begin{theorem}
    If $\gamma :\Gamma$ and $\Gamma \ts e : \tau$ and $\gamma,e \red
    v,T$ then $v : \tau$ and $\Gamma \ts T : \tau$.  If
    $\gamma:\Gamma$ and $\Gamma\ts T:\tau$ and $\gamma,T\trrun v$
    then $v:\tau$.
  \end{theorem}
\end{intronly}

Traces can be represented using pointers to share common
subexpressions; using this DAG representation, traces can be stored in
space polynomial in the input.  (This sharing happens automatically in
our implementation in Haskell.)
\begin{proposition}
  For a fixed $e$, if $\gamma,e \red v,T$ then the sizes of $v$ and of
  the DAG representation of $T$ are at most polynomial in $|\gamma|$.
\end{proposition}
\begin{proof}
  Most cases are straightforward.  The only non-trivial case is for
  comprehensions, where we need a stronger induction hypothesis: if
  $\gamma, x \in v, e \red^* v',\Theta$ then the sizes of $v'$ and of the DAG
  representation of $\Theta$ are at most polynomial in $|\gamma|$.  
\end{proof}
% Sharing happens automatically in our implementation since we use a
% functional language, Haskell, that never needs to copy the expressions appearing
% in traces; moreover, due to Haskell's lazy evaluation strategy we can
% often avoid explicitly constructing traces.  Further space
% optimizations or trace compression are possible, but beyond the scope
% of this paper.

Furthermore, traced evaluation produces a trace that replays to the
same value as the original expression run on the original
environment. We call this property \emph{consistency}.
\begin{proposition}[Consistency]\label{prop:consistency}
  If $\gamma,e \red v,T$ then $\gamma,T \trrun v$.
\end{proposition}

Finally, trace replay is faithful to ordinary evaluation in the
following sense: if $T$ is generated by running $e$ in $\gamma$
and we successfully replay $T$ on $\gamma'$ then
we obtain the same value (and same trace) as if we had rerun $e$ from
scratch in $\gamma'$, and vice versa:

\begin{proposition}[Fidelity]\label{prop:fidelity}
  If $\gamma,e \red v,T$, then for any $\gamma'$, $v'$ we have
  $\gamma',e \red v',T$ if and only if $\gamma',T \trrun v'$.
\end{proposition}
Observe that consistency is a special case of fidelity (with $\gamma'
= \gamma, v'=v$.  Moreover, the ``if'' direction of fidelity holds
even though replay can fail, because we require that $\gamma',e \red
v',T$ holds for the \emph{same trace $T$}.  If $\gamma',e \red v',T'$
is derivable but only for a different trace $T'$, then replay fails.

% LocalWords:  SQL multisets multiset nrc dependences subexpressions backslice
% LocalWords:  subtrace sim nondeterministic bijective TODO Umut COnsider versa
% LocalWords:  sequentializing unordered booleans prepends formalisms prepend
% LocalWords:  deterministically boolean Haskell Haskell's optimizations
% LocalWords:  recomputation

%%% Local Variables: 
%%% mode: latex
%%% TeX-
\section{Trace Slicing}
\label{sec:slicing}
The goal of the trace slicing algorithm we consider is to remove
information from a trace and input that is not needed to recompute a
part of the output. To accommodate these requirements, we introduce
traces and values with \emph{holes} and more generally, we consider
\emph{patterns} that represent relations on values capturing possible
changes.  In this section, we limit attention to pairs, and consider
slicing for a class of \emph{simple} patterns.  We consider records
and more expressive \emph{enriched} patterns in the next section.

We extend traces with holes $\tremp$
\begin{intronly}
  \begin{eqnarray*}
    T & ::= & \cdots \mid \tremp
  \end{eqnarray*}
\end{intronly}
and define a subtrace relation $\sqleq$ that is essentially a
syntactic precongruence on traces and trace sets such that $\vhole
\sqleq T$ holds and $\Theta \subseteq \Theta'$ implies $\Theta \sqleq
\Theta'$.  (The definition of $\sqleq$ is shown in full in the
companion technical report.) Intuitively, holes denote parts of traces we do not care
about, and we can think of a trace $T$ with holes as standing for a
set of possible complete traces $\{T' \mid T \sqleq T'\}$ representing
different ways of filling in the holes.

%  To formulate and prove the appropriate completeness results for NRC
% traces and slicing, we need to extend the subtrace and $\eqat{p}$
% relations.  The key rules in these extensions are shown in
% \figreftwo{nrc-subtrace}{nrc-sim}; many routine rules are omitted.
%  We show how to extend the backward and forward slicing techniques
% developed in \secref{lambda-slicing} to NRC.  
% For backward slicing, in particular, we encounter a number of
% interesting issues.  

The syntax of \emph{simple patterns} $p$, \emph{set
patterns}  $\spat$, and \emph{pattern environments} $\rho$ is:
\begin{eqnarray*}
  p &::=& \vhole \mid \vany \mid \vc \mid \vpair{p_1}{p_2} \mid
 \spat\\
\spat &::=& \emptyset \mid \vset{\lbl_1.p_1,\ldots,\lbl_n.p_n}\\
\rho &::=& [x_1 \mapsto p_1,\ldots,x_n \mapsto p_n]
\end{eqnarray*}
Essentially, a pattern is a value with holes in some positions.  The
meaning of each pattern is defined through a relation $\simat{p}$
that says when two values are equivalent with respect to a pattern.
This relation is defined in Figure~\ref{fig:nrc-sim}.
A hole $\vhole$ indicates that the part of the value is unimportant,
that is, for slicing purposes we don't care about that part of the
result.  Its associated relation $\simat{\vhole}$ relates any two values.  An
\emph{identity pattern} $\vany$ is similar to a hole: it says that the
value is important but its exact value is not specified, and its associated
relation $\simat{\vany}$ is the identity relation on values.  Complete set patterns
$\{\lbl_1.p_1,\ldots,\lbl_n.p_n\}$ specify the labels and patterns for
all of the elements of a set; that is, such a pattern relates only
sets that have exactly the labeled elements specified and whose
corresponding values match according to the corresponding patterns.
% A
% partial set pattern tagged with $\subseteq$  says that there may be some additional values;
% that is, it relates any two sets that have at least the specified
% labels (and maybe others), such that the labeled elements are
% equivalent with respect to the patterns.

%\input{fig-nrc-subtrace}

\begin{figure}[tb]
\[
  \begin{array}{rcl}
\vhole \sqcup p =  p \sqcup \vhole &=& p\\
\vany \sqcup p =  p \sqcup \vany &=& p[\vany/\vhole]\\
\vc \sqcup \vc &=& 
\vc\\
% \vrec{\overline{A_i:p_i}} \sqcup  \vrec{\overline{A_i:p_i'}} &=&
% \vrec{\overline{A_i:p_i\sqcup p_i'}}\\
\vpair{p_1}{p_2} \sqcup \vpair{p_1'}{p_2'} &=& \vpair{p_1\sqcup
  p_1'}{p_2\sqcup p_2'}\\
   \vset{\overline{\lbl_i.p_i}}\sqcup    \vset{\overline{\lbl_i.p_i'}} &=&    \vset{\overline{\lbl_i.p_i\sqcup p_i'}}
  \end{array}
\]\text{where:}\[
\begin{array}{rcl}
  \vany[\vany/\vhole]  = \vhole[\vany/\vhole] &=& \vany\\
 \vc[\vany/\vhole] &=& \vc\\
\vpair{p_1}{p_2}[\vany/\vhole] &=& \vpair{p_1[\vany/\vhole]}{p_2[\vany/\vhole]}\\
%  \vrec{\overline{A_i:p_i}}[\vany/\vhole] &=& \vrec{\overline{A_i:p_i[\vany/\vhole]}}\\
   \vset{\overline{\lbl_i.p_i}}[\vany/\vhole] &=& \vset{\overline{\lbl_i.p_i[\vany/\vhole]}}
\end{array}
\]
  \caption{Least upper bound for simple patterns}
  \label{fig:nrc-lub}
\end{figure}
\begin{figure}
\fbox{$v \simat{p} v'$}
\vspace{-1ex}
\begin{smathpar}
\inferrule*
{\strut}
{v \simat{\vhole} v'}
\and
\inferrule*
{\strut}
{v \simat{\vany} v}
\and
\inferrule*{
\strut
}{
\vc \simat{\vc} \vc
}
\and
\inferrule*
{v_1 \simat{p_1} v_1'\\
v_1 \simat{p_2} v_2'
}
{\kwpair{v_1}{v_2} \simat{\kwpair{p_1}{p_2}} \kwpair{v_1'}{v_2'}}
\and
\inferrule*{
v_i \simat{p_i} v_i' \quad (i \in \{1,\ldots,n\})
}{
\{\lbl_1.v_1,\ldots,\lbl_n.v_n\}
\simat{\{\lbl_1.p_1,\ldots,\lbl_n.p_n\}} 
\{\lbl_1.v_1',\ldots,\lbl_n.v_n'\}
}
\end{smathpar}
\[\gamma \eqat{\rho} \gamma' \iff \forall x \in dom(\rho).~\gamma(x)
\simat{\rho(x)} \gamma'(x)\]
\caption{Simple pattern equivalence.}
\label{fig:nrc-sim}
\end{figure}

We define the union of two set patterns as
$\vset{\overline{\lbl_i.p_i}} \semu \vset{\overline{\lbl_i'.p_i'}} =
\vset{\overline{\lbl_i.p_i}, \overline{\lbl_i'.p_i'}}$ provided their
domains $\vec{\lbl_i}$ and $\vec{\lbl_i'}$ are prefix-disjoint.  We
define a (partial) least upper bound operation on patterns $p \sqcup
p'$ such that for any $v,v'$ we have $v\eqat{p \sqcup p'}v' $ if and
only if $ v\eqat{p}v' $ and $v\eqat{p'}v'$; the full definition is
shown in \figref{nrc-lub}.  We define the partial ordering $p \sqleq
p'$ as $p \sqcup p' = p'$.  We say that a value $v$ \emph{matches}
pattern $p$ if $p \sqleq v$.  Observe that this implies $v \simat{p}
v$.  We extend the $\sqcup$ and $\sqleq$ operations to pattern
environments $\rho$ pointwise, that is, $(\rho\sqcup \rho')(x) =
\rho(x) \sqcup \rho(x')$.

We define several additional operations on patterns that are needed
for the slicing algorithm. Consider the following
\emph{singleton extraction} operation $p.\epsilon$ and \emph{label
  projection} operation $p[\lbl]$:
\[\begin{array}{rcll}
(\{\epsilon.p\}).\epsilon &=& p&\vhole.\epsilon = \vhole\quad \vany.\epsilon =\vany\\
% \kwrec{A_1:p_1,\ldots,A_n:p_n}.A_i &=& p_i\\
% \vhole.A &=& \vhole \qquad \vany.A = \vany\\
\spat[\lbl] &=&  \{\lbl'.v \mid \lbl.\lbl'.v \in \spat\}& \vhole[\lbl] = \vhole \quad\vany[\lbl]=\vany
\end{array}
\]
These operations are only used for the above cases; they have no
effect on constant or pair patterns.  For sets,
$p[\lbl]$ extracts the subset of $p$ whose labels start with $\lbl$,
truncating the initial prefix $\lbl$, while if $p$ is $\vhole$ or
$\vany$ then again $p[\lbl]$ returns the same kind of hole.  Moreover,
$\dom(p[\lbl])$ is a prefix code if $\dom(p)$ is.

Suppose $L$ is a prefix code.  We define \emph{restriction} of a
set pattern $p$ to $L$ as follows:
\[\spat|_{L} = \{\lbl.\lbl'.p \in \spat \mid \lbl \in L\} \qquad \vhole|_L = \vhole \qquad \vany |_L = \vany\]
It is easy to see that $\dom(p|_L) \subseteq \dom(p) $ so $\dom(p|_L)$
is a prefix code if $\dom(p)$ is, so this operation is well-defined on
collections:
\begin{lemma}
  If $\spat$ is prefix-labeled and $L$ is a prefix code then
  $\spat[\lbl]$ and $\spat|_L$ are prefix-labeled.
\end{lemma} 
\begin{intronly}
 We show other properties
of patterns in  Appendix~\ref{app:pattern-proofs}.
\end{intronly}

\paragraph*{Backward Slicing}

The rules for backward slicing are given in \figref{nrc-backslice}.
The judgment $p, T \trbslice\rho, T'$ slices trace $T$ with respect to
a pattern $p$ to yield the slice $T'$ and sliced input environment
$\rho$.  The sliced input environment records what parts of the input
are needed to produce $p$; this is needed for slicing operations such
as let-binding or comprehensions.  The main new ideas are in the rules
for collection operations, particularly comprehensions. The slicing
rules $\textsc{SConst}$, $\textsc{SPrim}$, $\textsc{SVar}$,
$\textsc{SLet}$, $\textsc{SPair}$, $\textsc{SProj}_i$, and
$\textsc{SIf}$ follow essentially the same idea as in our previous
work~\cite{acar13jcs,perera12icfp}.  We focus discussion on the new
cases, but we review the key ideas for these operations here in order
to make the presentation self-contained.

\begin{figure*}[tb!]
\fbox{$p, T \trbslice \rho, S$}
\vspace{-2ex}
\begin{smathpar}
\inferrule*[right=SHole]
{\strut}
{
\vhole,T \trbslice [],\vhole
}
\and
\inferrule*[right=SConst]
{
\strut
}
{
p,\trc \trbslice [],\trc
}
\and
\inferrule*[right=SPrim]
{
\vany, T_1 \trbslice \rho_1,S_1\\
\cdots\\
\vany,T_n \trbslice\rho_n, S_n
}
{
p,\kwf(T_1,\ldots,T_n) \trbslice \rho_1 \sqcup \cdots \sqcup
\rho_n, \trf(S_1,\ldots,S_n)
}
\and
\inferrule*[right=SVar]
{\strut
}
{
p,x \trbslice[x\mapsto p], x
}
\and
\inferrule*[right=SLet]
{
p_2,T_2 \trbslice \rho_2[x\mapsto p_1],S_2\\
p_1,T_1 \trbslice \rho_1, S_1
}
{
p_2,\trlet{x}{T_1}{T_2} \trbslice \rho_1 \sqcup \rho_2, \trlet{x}{S_1}{S_2}
}
\and
\inferrule*[right=SPair]
{
p_1,T_1\trbslice \rho_1,S_1\\
p_2,T_2\trbslice \rho_2,S_2
}
{
(p_1,p_2),\trpair{T_1}{T_2} \trbslice \rho_1 \sqcup \rho_2, \trpair{S_1}{S_2}
}
\and
\inferrule*[right=SProj$_1$]{
\vpair{p}{\vhole},T \trbslice \rho,S
}{
p,\trfield{T}{\#_1} \trbslice \rho,\trfield{S}{\#_1}
}
\and
\inferrule*[right=SProj$_2$]{
\vpair{\vhole}{p},T \trbslice \rho,S
}{
p,\trfield{T}{\#_2} \trbslice \rho,\trfield{S}{\#_2}
}
\and
% \inferrule*
% {
% p.A_1,T_1 \trbslice \rho_1,S_1\\
% \cdots\\
% p.A_n,T_n \trbslice \rho_n,S_n
% }
% {
% p , \trrec{A_1{:}T_1,\ldots,A_n{:}T_n}
% \trbslice \rho_1 \sqcup \cdots \sqcup \rho_n,
% \kwrec{A_1{:}S_1,\ldots,A_n{:}S_n}
% }
%
% \and
% %
% \inferrule*{
% \vrec{A:p,\vhole},T \trbslice \rho,S
% }{
% p,\trfield{T}{A} \trbslice \rho,\trfield{S}{A}
% }
%
%\and
%
\inferrule*[right=SIf]
{
p,T' \trbslice \rho',S' \\
b,T \trbslice \rho,S
}
{
p,\trif{T}{e_1}{e_2}{b}{T'} \trbslice \rho'\sqcup \rho, \trif{S}{e_1}{e_2}{b}{S'}
}
\and
\inferrule*[right=SEmpty]
{
\strut
}
{
p, \emptyset \trbslice [],\emptyset
}
\and
\inferrule*[right=SSng]
{
p.\epsilon, T \trbslice \rho,S
}
{
p, \trsets{T} \trbslice \rho,\trsets{S}
}
% \and
%
\and
\inferrule*[right=SUnion]
{
p[1], T_1 \trbslice \rho_1, S_1 \\
p[2], T_2 \trbslice \rho_2, S_2 
}
{
 p, \trsetu{T_1}{T_2} \trbslice \rho_1 \sqcup \rho_2 , \trsetu{S_1}{S_2}
}
\and
\inferrule*[right=SComp]{
 p, x.\Theta \trbslice^* \rho', \Theta', p'\\
p', T \trbslice \rho,S
}{
p,\trcomp{e}{x}{T}{\Theta} \trbslice \rho \sqcup \rho', \trcomp{e}{x}{S}{\Theta'}
}
\and
\inferrule*[right=SEmptyP]
{
 \vany, T \trbslice\rho, S
}
{
 p, \trsetsum{T} \trbslice\rho, \trsetsum{S}
}
\and
\inferrule*[right=SSum]
{
 \vany, T \trbslice \rho,S
}
{
 p, \trsetise{T} \trbslice \rho,\trsetise{S}
}
\and
\inferrule*[right=SDiamond]{
\vany,T_1 \trbslice \rho_1,S_1\\
\vany,T_2 \trbslice \rho_2,S_2
}{
\vany, (T_1,T_2)\trbslice \rho_1 \sqcup \rho_2, (S_1,S_2)
}
\end{smathpar}
\fbox{$ p, x.\Theta \trbslice^* \rho, \Theta_0,p_0$}
\vspace{-5ex}
\begin{smathpar}
% \inferrule*{
% \strut
% }{
% \vhole, x.\Theta \trbslice^* [], \emptyset, \vhole
% }
% \and
% \inferrule*{
% \strut
% }{
% p, x.\Theta \trbslice^* [], \emptyset, \emptyset
% }
% \and
% \inferrule*{
% \strut
% }{
% x \in \emptyset^\subseteq, \Theta \trbslice^* [], \emptyset, \emptyset^\subseteq
% }
% \and
% \inferrule*{
% \{\lbl'.p\}^R,\Theta(\lbl) \trbslice \rho[x\mapsto p'],S
% }{
% x \in \{\lbl.\lbl'.v\}^R,\Theta\trbslice^* \rho,\{\lbl.S\},\{\lbl.p'\}^R
% }
% \and
% \inferrule*{
% x \in p_1,\Theta \trbslice^* \rho_1,\Theta_1,p_1'\\
% x \in p_2,\Theta \trbslice^* \rho_2,\Theta_2,p_2'
% }{
% x \in p_1 \semu p_2,\Theta \trbslice^* \rho_1 \sqcup\rho_2, \Theta_1 \sqcup
% \Theta_2, p_1' \semu p_2'
% }
%
%\\
%
% \inferrule*
% {
% \strut
% }
% {
% \vany, \emptyset \trbslice^* [], \emptyset, \emptyset^=
% }
% %
% \and
% %
% \inferrule*
% {
% \vany, T \trbslice \rho[x\mapsto p],S
% }
% {
% x \in \vany, \{\lbl.T\} \trbslice^* \rho, \{\lbl.S\}, \{\lbl.p\}^=
% }
%
% \and
% %
% \inferrule*
% {
% \vany, x.\Theta_1 \trbslice^* \rho_1,\Theta_1',p_1\\
% \vany, x.\Theta_2 \trbslice^* \rho_2,\Theta_2',p_2
% }
% {
% \vany, x.\Theta_1 \semu \Theta_2 \trbslice^* \rho_1 \sqcup \rho_2,
% \Theta_1' \semu \Theta_2', p_1 \semu p_2
% }
% %
\inferrule*[right=SEmpty$^*$]
{
\strut
}
{
\emptyset, x.\emptyset \trbslice^* [], \emptyset, \emptyset
}
\and
\inferrule*[right=SSng$^*$]
{
p[\lbl], T \trbslice \rho[x\mapsto p_0],S
}
{
p, x.\{\lbl.T\} \trbslice^* \rho, \{\lbl.S\}, \{\lbl.p_0\}
}
\and
\inferrule*[right=SUnion$^*$]
{
p|_{\dom(\Theta_1)}, x.\Theta_1 \trbslice^* \rho_1,\Theta_1',p_1\\
p|_{\dom(\Theta_2)}, x.\Theta_2 \trbslice^* \rho_2,\Theta_2',p_2
}
{
p, x.\Theta_1 \semu \Theta_2 \trbslice^* \rho_1 \sqcup \rho_2,
\Theta_1' \semu \Theta_2', p_1 \semu p_2
}
\end{smathpar}
\caption{Backward trace slicing.}
\label{fig:nrc-backslice}
\end{figure*}

The rules for collections use labels and set pattern operations to
effectively undo the evaluation of the set pattern.  The rule $\textsc{SEmpty}$
is essentially the same as the constant rule.  The rule $\textsc{SSng}$
uses the singleton extraction operation $p.\epsilon$ to obtain a pattern
describing the single element of a singleton set value matching $p$.
The rule $\textsc{SUnion}$ uses the two projections $p[1]$ and $p[2]$ to obtain
the patterns describing the subsets obtained from the first and second
subtraces in the union pattern, respectively.

The rule $\textsc{SComp}$ uses a similar idea to let-binding.  We
slice the trace set $\Theta$ using an auxiliary judgment $p,x.\Theta
\trbslice^* \rho,\Theta_0,p_0$, obtaining a sliced input environment,
sliced trace set $\Theta_0$, and pattern $p_0$ describing the set of
values to which $x$ was bound.  We then use $p_0$ to slice backwards
through the subtrace $T$ that constructed the set.  The auxiliary
slicing judgment for trace sets has three rules: a trivial rule $\textsc{SEmpty}^*$ when
the set is empty, rule $\textsc{SSng}^*$ that uses label projection $p[\lbl]$ 
to handle a singleton trace set, and a rule $\textsc{SUnion}^*$ handling larger trace sets
by decomposing them into subsets.  This is essentially a structural
recursion over the trace set, and is deterministic even though the
rules can be used to decompose the trace sets in many different ways,
because $\semu$ is associative.
Rule $\textsc{SUnion}^*$ also requires that we restrict the set pattern to match the
domains of the corresponding trace patterns.

The slicing rules $\textsc{SSum}$ and $\textsc{SEmptyP}$ follow the
same idea as for primitive operations at base type: we require that
the whole set value be preserved exactly.  As discussed by Perera et
al.~\cite{perera12icfp} with primitive operations, this is a potential
source of overapproximation, since (for example) for an emptiness
test, all we really need is to preserve the number of elements in the
set, not their values.  The last rule shows how to slice pair patterns
when the pattern is $\vany$: we slice both of the subtraces by $\vany$
and combine the results.

%OLD
% \begin{theorem}[Completeness for NRC Traces]\label{thm:completeness-nrc-trace}
%   If $e \red v,T$, then for any $\theta$, $v'$ with $\theta(e) \red
%   v'$, we have $\theta(T) \trrun v''$ for some $v'' \cong v'$.
% \end{theorem}

% And here's the version I actually proved: 

% \begin{theorem}[Completeness for NRC Traces]
%   If $e \red v,T$, then for any $\theta$, $h$, $v'$ with $\theta(h(e)) \red v'$, there exists $h_1 \supseteq h$ such that $\theta (h(T)) \trrun h_1^{-1}v'$. 
% \end{theorem}

% We prove correctness for forward and backward slices, extending
% $\BSlice(T,p)$ and $\FSlice(T,\lbl)$ to NRC in the obvious ways.
% Analogues of \thmreftwo{completeness-backward}{completeness-forward}
% follow using \thmref{completeness-nrc-trace}.
% % \begin{theorem}[Correctness for NRC Backward Slicing]
%   Suppose $T \trrun v_0$ and $p, T \trbslice T'$.  Then $T' \in
%   \BSlice(T,p)$. 
% % Then $T' \trsub T$.
% %   Also, for any $\theta$, $v$ such that $\theta(T) \trrun v$, we have that there exists $v'$ such that $\theta(T') \trrun v'$ and $v \eqat{p} v'$.
% \end{theorem}

It is straightforward to show that the slicing algorithm is
well-defined for consistent traces.  That is, if $\gamma,T \trrun v$
and $p\sqleq v$ then there exists $\rho,S$ such that $p,T \trbslice
\rho,S$ holds, where $\rho \sqleq \gamma$ and $S \sqleq T$.  We defer
the correctness theorem for slicing using simple patterns to the end of the next section,
since it is a special case
of correctness for slicing using enriched patterns.

% It is important to note that the algorithm does not guarantee
% recomputability under \emph{all} changes to the input.  This is often
% the case, but is not guaranteed, because slicing may discard
% information that would be needed in another control flow branch.  For
% example, if $y=1$ then slicing
% %
%   \[\kwlet{x}{(1,2)}{\trifthen{y=1}{x.\#_1}{x.\#_2}{x.\#_1}}\]
% %
% with respect to its full output $1$ yields 
% %
% \[\trlet{x}{(1,\vhole)}{\trifthen{y=1}{x.\#_1}{x.\#_2}{x.\#_1}}\]
% %
%   and if we rerun this with the hole filled with some other value than
%   $2$, then the result will not match that of the original trace.
%   Obtaining this stronger property would require analysis of non-taken
%   branches or non-evaluated comprehension bodies to identify all parts
%   of the trace that could have contributed to the output at $p$.  This
%   is challenging, and moreover would make slices larger and (likely)
%   less comprehensible, by considering counterfactual situations.  We
%   leave this problem for future work.

% \begin{theorem}[Correctness for NRC Forward Slicing]
%   Suppose $T \trrun v_0$ and $\lbl$ is a free label in $T$ and 
%   $\lbl, T \trfslice T'$.  Then $T' \in\FSlice(T,\lbl)$.
% % Then $T' \trsub T$ and 
% %   for any values $v$, $v'$, if
% %   $\subst{\lbl.v'}{\lbl}{T} \trrun v$ 
% %   then $\subst{\lbl.v'}{\lbl}{T'} \trrun v$.
%  \end{theorem}

Continuing our running example, consider the pattern $p =
\{[r_2].\kwrec{A{:}\vhole,B{:}8},[r_3].\vhole\}$.  The
slice of $T$ with respect to this pattern is of the form:
\[
\begin{array}{ll}
T' =   \trcomp{\_}{x}{R}{\{&[r_1].\vhole,\\
& [r_2].\trifthen{x.B=3}{\_}{\_}{\{\_\}},\\
&[r_3].\vhole~~ \}}\\
\end{array}
\]
and the slice of $R$ is $R' =
\{[r_1].\vhole,[r_2].\kwrec{A{:}\vhole,B{:}3,C{:}8},[r_3].\vhole\}$.
(That is, $R'$ is the value of $\rho(R)$, where $\rho$ is the pattern
environment produced by slicing $T$ with respect to $p$.)
Observe that the value of $A$ is not needed but the value of $B$ must
remain $3$ in order to preserve the control flow behavior of the trace
on $r_2$.  The holes indicate that changes to $r_1$ and $r_3$ in the
input cannot affect $r_2$.  However, a trace
matching $T'$ cannot be replayed if any of $r_1,r_2,r_3$ are deleted
from the input or if the $A$ field is removed from an input record,
because the replay rules require all of the labels mentioned in
collections or records in $T'$ to be present.
We now turn our attention to enriched patterns, which mitigate these
drawbacks.

% LocalWords:  subtrace nrc lub rcll backslice recomputation subtraces
% LocalWords:  overapproximation recomputability counterfactual

%%% Local Variables: 
%%% mode: latex
%%% TeX-master: "main"
%%% End: 

\section{Enriched Patterns}
\label{sec:patterns}

So far we have considered only complete set patterns of the form
$\set{\lbl_1.p_1,\ldots,\lbl_n.p_n}$.  These patterns relate pairs of values
that have exactly $n$ elements labeled $\lbl_1,\ldots,\lbl_n$, each of
which  matches $p_1,\ldots,p_n$ respectively.  

This is awkward, as we already can observe in our running example
above: to obtain a slice describing how one record was computed, we
need to use a set pattern that lists all of the indexes in the
output.  Moreover,
the slice with respect to such a pattern may also include labeled
subtraces explaining why the other elements exist in the
output (and no others).
This information seems intuitively irrelevant to the value at
$\lbl_1$, and can be a major overhead if the collection is large.
We also have considered only binary pairs;  % Of course, arbitrary pairs
% and records can be simulated using binary pairs alone, but doing so
% is awkward from the point of view of slicing, since we must translate
% a pattern specified in terms of records to one specified using pairs,
% slice the trace using the pair slicing rules, and finally translate
% the traced slice back to record syntax.  
it would be more convenient to
support record patterns directly.

In this section we sketch how to  enrich the language of patterns to allow for
\emph{partial set} and \emph{partial record patterns}, as follows:
\begin{eqnarray*}
  p &::=& \vhole \mid \vany \mid \vc \mid
 \spat \mid \rpat\\
\rpat &::=& \vrec{} \mid \vrec{\overline{A_i:p_i}} \mid
\vrec{\overline{A_i:p_i};\vhole} \mid
\vrec{\overline{A_i:p_i};\vany}\\
\spat &::=& \emptyset \mid \vset{\overline{\lbl_i.p_i}} \mid \vset{\overline{\lbl_i.p_i}} \udot \vhole
\mid \vset{\overline{\lbl_i.p_i}} \udot \vany
\end{eqnarray*}
Record patterns are of the form $\vrec{\overline{A_i:p_i}}$, listing
the fields and the patterns they must match, possibly followed by
$\vhole$ or $\vany$, which the remainder of the record must match.
The pattern $\{\overline{\lbl_i.p_i}\} \udot \vhole$ stands for a set
with labeled elements matching patterns $p_1,\ldots,p_n$, plus some
additional elements whose values we don't care about.  For example, we
can use the pattern $\{\lbl_1.p_1\} \udot \vhole$ to express interest
in why element $\lbl_1$ matches $p_1$, when we don't care about the
rest of the set.  The second partial pattern, $
\{\overline{\lbl_i.p_i}\} \udot \vany$, has similar behavior, but it
says that the rest of the sets being considered must be equal.  For
example, $\{\lbl.\vhole\}\udot \vany$ says that the two sets are equal
except possibly at label $\lbl$.
This pattern is needed mainly in order to ensure that we can define a
least upper bound on enriched patterns, since we cannot express
$(\{\lbl_1.v_1,\ldots,\lbl_n.p_n\} \udot \vhole) \sqcup \vany$
otherwise.

We define $\dom(\{\overline{\lbl_i{.}p_i}\}\udot \vhole) =
\dom(\{\overline{\lbl_i{.}p_i}\}\udot \vany) =
\{\lbl_1,\ldots,\lbl_n\}$. Disjoint union of enriched patterns
$\spat \semu \spat'$ is defined only if the domains are
prefix-disjoint, so that the labels of the result still form a prefix
code; this operation is defined in 
\figref{nrc-pattern-union}.  
\begin{figure}[t]
  \[\small
  \begin{array}{rcl}
    \vany \semu \vhole = \vhole \semu \vany =  \vhole \semu \vhole &=&
    \vhole\\
    \vany \semu \vany &=& \vany\\
    \spat\semu \emptyset = \emptyset \semu \spat &=& \spat\\
    \{\overline{\lbl_i.p_i}\} \semu \vhole = \vhole \semu \{\overline{\lbl_i.p_i}\} &=& \{\overline{\lbl_i.p_i}\} \udot \vhole\\
    \{\overline{\lbl_i.p_i}\} \semu \vany = \vany \semu \{\overline{\lbl_i.p_i}\} &=& \{\overline{\lbl_i.p_i}\} \udot \vany\\
    \{\overline{\lbl_i.p_i}\} \semu (\{\overline{\lbl_j'.p_j'}\}\udot \vhole) &=& (\{\overline{\lbl_i.p_i}\} \semu \{\overline{\lbl_j'.p_j'}\}) \udot \vhole\\
    \{\overline{\lbl_i.p_i}\} \semu (\{\overline{\lbl_j'.p_j'}\} \udot \vany) &=&  (\{\overline{\lbl_i.p_i}\} \semu \{\overline{\lbl_j'.p_j'}\}) \udot \vany\\
    (\{\overline{\lbl_i.p_i}\}\udot \vhole) \semu (\{\overline{\lbl_j'.p_j'}\}\udot \vhole) &=& (\{\overline{\lbl_i.p_i}\} \semu \{\overline{\lbl_j'.p_j'}\}) \udot \vhole\\
    (\{\overline{\lbl_i.p_i}\}\udot \vany) \semu (\{\overline{\lbl_j'.p_j'}\}\udot \vhole) & =&  (\{\overline{\lbl_i.p_i}\} \semu \{\overline{\lbl_j'.p_j'}\}) \udot \vhole\\
    (\{\overline{\lbl_i.p_i}\}\udot \vhole) \semu (\{\overline{\lbl_j'.p_j'}\}\udot \vany) &=& (\{\overline{\lbl_i.p_i}\} \semu \{\overline{\lbl_j'.p_j'}\}) \udot \vhole\\
    (\{\overline{\lbl_i.p_i}\}\udot \vany) \semu
    (\{\overline{\lbl_j'.p_j'}\}\udot \vany) &=&
    (\{\overline{\lbl_i.p_i}\} \semu \{\overline{\lbl_j'.p_j'}\}) \udot
    \vany
  \end{array}\]
  
\caption{Enriched pattern union}\label{fig:nrc-pattern-union}
\end{figure}

% The pattern equivalence relation $\eqat{p}$, least upper bound $\sqcup$, singleton
% extraction $p.\epsilon$, label projection $p[\lbl]$, and restriction
% $p|_L$ operations are also extended to deal with enriched patterns,
% and we define a projection operation
% on record patterns $p.A$.  

%Appendix~\ref{app:patterns}.

We now extend the definitions of $p[\vany/\vhole]$ and
$\sqcup$ to account for extended patterns.  We extend the $[\vany/\vhole]$
substitution operation as follows:
\[\small\begin{array}{rcl}
  \vrec{\overline{A_i:p_i}}[\vany/\vhole] &=&
  \vrec{\overline{A_i:p_i[\vany/\vhole]}}\\
  \vrec{\overline{A_i:p_i};\vhole}[\vany/\vhole] &=&
  \vrec{\overline{A_i:p_i[\vany/\vhole];}\vany}\\
   \vrec{\overline{A_i:p_i};\vany}[\vany/\vhole] &=&
  \vrec{\overline{A_i:p_i[\vany/\vhole]};\vany}\\
 (\{\overline{\lbl_i.p_i}\} \udot \vhole)[\vany/\vhole] &=&
  \{\overline{\lbl_i.p_i[\vany/\vhole]}\} \udot \vany\\
  (\{\overline{\lbl_i.p_i}\} \udot \vany)[\vany/\vhole] &=&
  \{\overline{\lbl_i.p_i[\vany/\vhole]}\} \udot \vany
\end{array}\]

\begin{figure*}[tb]
\small
  \begin{eqnarray*}
\vset{\overline{\lbl_i.p_i},\overline{\lbl_j'.q_j}} \sqcup  (\vset{\overline{\lbl_i.p_i'}}\udot \vhole) &=&  \vset{\overline{\lbl_i.p_i\sqcup p_i'},\overline{\lbl_j'.q_j}}\\
\vset{\overline{\lbl_i.p_i},\overline{\lbl_j'.q_j}} \sqcup  (\vset{\overline{\lbl_i.p_i'}}\udot \vany) &=&  \vset{\overline{\lbl_i.p_i\sqcup p_i'},\overline{\lbl_j'.q_j[\vany/\vhole]}}\\
(\vset{\overline{\lbl_i.p_i},\overline{\lbl_j'.q_j}} \udot \vhole) \sqcup
(\vset{\overline{\lbl_i.p_i'},\overline{\lbl_k''.r_k}}\udot \vhole) &=&
\vset{\overline{\lbl_i.p_i\sqcup p_i'},\overline{\lbl_j'.q_j},\overline{\lbl_k''.r_k} } \udot \vhole\\
(\vset{\overline{\lbl_i.p_i},\overline{\lbl_j'.q_j}} \udot \vhole) \sqcup
(\vset{\overline{\lbl_i.p_i'},\overline{\lbl_k''.r_k}}\udot \vany) &=&
\vset{\overline{\lbl_i.p_i\sqcup
    p_i'},\overline{\lbl_j'.q_j[\vany/\vhole]},\overline{\lbl_k''.r_k}} \udot \vany\\
(\vset{\overline{\lbl_i.p_i},\overline{\lbl_j'.q_j}} \udot \vany) \sqcup
(\vset{\overline{\lbl_i.p_i'},\overline{\lbl_k''.r_k}}\udot \vany) &=&
\vset{\overline{\lbl_i.p_i\sqcup
    p_i'},\overline{\lbl_j'.q_j[\vany/\vhole]},\overline{\lbl_k''.r_k[\vany/\vhole]}} \udot \vany
  \end{eqnarray*}
   \begin{eqnarray*}
\vrec{\overline{A_i:p_i},\overline{B_j:q_j}} \sqcup  (\vrec{\overline{A_i:p_i'};\vhole}) &=&  \vrec{\overline{A_i:p_i\sqcup p_i'},\overline{B_j:q_j}}\\
\vrec{\overline{A_i:p_i},\overline{B_j:q_j}} \sqcup  (\vrec{\overline{A_i:p_i'};\vany}) &=&  \vrec{\overline{A_i:p_i\sqcup p_i'},\overline{B_j:q_j[\vany/\vhole]}}\\
(\vrec{\overline{A_i:p_i},\overline{B_j:q_j}; \vhole}) \sqcup
(\vrec{\overline{A_i:p_i'},\overline{C_k:r_k};\vhole}) &=&
\vrec{\overline{A_i:p_i\sqcup p_i'},\overline{B_j:q_j},\overline{C_k:r_k} ; \vhole}\\
(\vrec{\overline{A_i:p_i},\overline{B_j:q_j};\vhole}) \sqcup
(\vrec{\overline{A_i:p_i'},\overline{C_k:r_k};\vany}) &=&
\vrec{\overline{A_i:p_i\sqcup
    p_i'},\overline{B_j:q_j[\vany/\vhole]},\overline{C_k:r_k}; \vany}\\
(\vrec{\overline{A_i:p_i},\overline{B_j:q_j}; \vany}) \sqcup
(\vrec{\overline{A_i:p_i'},\overline{C_k:r_k}; \vany}) &=&
\vrec{\overline{A_i:p_i\sqcup
    p_i'},\overline{B_j:q_j[\vany/\vhole]},\overline{C_k:r_k[\vany/\vhole]}; \vany}
  \end{eqnarray*} \caption{Least upper bound for enriched patterns
    (excluding some symmetric cases).}
  \label{fig:nrc-lub-extended}
\end{figure*}

We handle the additional cases of the $\sqcup$ operation in
\figref{nrc-lub-extended}, and extend the $\simat{p}$ relation as
shown in \figref{nrc-sim-extended}.  Note that taking the least
upper bound of a partial pattern with a complete pattern yields a
complete pattern, while taking the least upper bound of partial
patterns involving $\vany$ again relies on the $[\vany/\vhole]$
substitution operation.

\begin{figure}[tb]
  \begin{smathpar}
\inferrule*
{v_1 \simat{p_1} v_1' \\
\cdots \\
v_n \simat{p_n} v_n'
}
{\kwrec{\overline{A_i:v_i}} \simat{\kwrec{\overline{A_i:p_i}}}
  \kwrec{\overline{A_i:v_i'}}
}
\and
\inferrule*
{v_1 \simat{p_1} v_1' \\
\cdots \\
v_n \simat{p_n} v_n'
}
{\kwrec{\overline{A_i:v_i},\overline{B_j{:}w_j}}
    \simat{\kwrec{\overline{A_i:p_i};\vhole}}
  \kwrec{\overline{A_i:v_i'}, \overline{C_k{:}w_k'}}
}
\and
\inferrule*
{v_1 \simat{p_1} v_1'\\
\cdots\\
v_n \simat{p_n} v_n'
}
{\kwrec{\overline{A_i:v_i},\overline{B_j{:}w_j}} \simat{\kwrec{\overline{A_i:p_i};\vany}}
  \kwrec{\overline{A_i:v_i'},\overline{B_j{:}w_j}}
}
\and
    \inferrule*{
v_1 \simat{p_1} v_1'\\
\cdots\\
v_n \simat{p_n} v_n'
 }{
\set{\overline{\lbl_i.v_i}}\semu v \simat{\set{\overline{\lbl_i.p_i}} \udot \vhole} \set{\overline{\lbl_i.v_i'}} \semu v'
}
\and
    \inferrule*{
v_1 \simat{p_1} v_1'\\
\cdots\\
v_n \simat{p_n} v_n'
 }{
\set{\overline{\lbl_i.v_i}}\semu v \simat{\set{\overline{\lbl_i.p_i}} \udot \vany} \set{\overline{\lbl_i.v_i'}} \semu v
}
  \end{smathpar}
  \caption{Enriched pattern equivalence}
  \label{fig:nrc-sim-extended}
\end{figure}

We extend the singleton extraction $p.\epsilon$, label projection
$p[\lbl]$, and restriction $p|_L$ operations on set patterns
as follows:
\[\small\begin{array}{rcl}
  (\{\epsilon.p\} \udot \vany) .\epsilon = (\{\epsilon.p\} \udot
  \vhole).\epsilon &=& p \\
\lbl \cdot (\vset{\overline{\lbl_i{.}p_i}} \udot \vhole) &=& (\lbl
\cdot \vset{\overline{\lbl_i{.}p_i}} )\semu \vhole\\
\lbl \cdot (\vset{\overline{\lbl_i{.}p_i}} \udot \vany) &=& (\lbl
\cdot \vset{\overline{\lbl_i{.}p_i}}) \semu \vany\\
(\vset{\overline{\lbl_i{.}p_i}} \udot \vhole)[\lbl] &=&
(\vset{\overline{\lbl_i{.}p_i}} [\lbl]) \semu \vhole\\
 (\vset{\overline{\lbl_i{.}p_i}} \udot \vany) [\lbl]&=& (
 \vset{\overline{\lbl_i{.}p_i}}[\lbl]) \semu \vany\\
(\vset{\overline{\lbl_i{.}p_i}} \udot \vhole)|_L &=&
(\vset{\overline{\lbl_i{.}p_i}}|_L) \semu \vhole\\
(\vset{\overline{\lbl_i{.}p_i}} \udot \vany)|_L &=&
(\vset{\overline{\lbl_i{.}p_i}}|_L) \semu \vany
\end{array}\]
Note that in many cases, we use the disjoint union operation 
$\semu$ on the right-hand side; this ensures, for example, that we never produce results of
the form $\emptyset \udot \vhole$ or $\emptyset \udot \vany$; these
are normalized to $\vhole$ and $\vany$ respectively, and this
normalization reduces the number of corner cases in the
slicing algorithm.

We define a record pattern projection operation $p.A$ as follows:
\[\small\begin{array}{rcl}
\kwrec{A_1{:}p_1,\ldots,A_n{:}p_n}.A_i &=& p_i \qquad \vhole.A = \vhole  \\
\kwrec{A_1{:}p_1,\ldots,A_n{:}p_n;\_}.A_i &=& p_i\qquad  \vany.A =\vany\\
\kwrec{A_1{:}p_1,\ldots,A_n{:}p_n;\vhole}.B &=& \vhole \quad (B \notin \{A_1,\ldots,A_n\})\\
\kwrec{A_1{:}p_1,\ldots,A_n{:}p_n;\vany}.B &=& \vany \quad (B \notin \{A_1,\ldots,A_n\})
\end{array}\]

We extend the slicing judgment to accommodate these new patterns in
Figure~\ref{fig:nrc-back-slice-extended}.  The rules $\textsc{SRec}$
and $\textsc{SProj}_A$ are similar to those for pairs, except that we
use the field projection operation in the case for a record trace, and
we use partial record patterns $\vrec{A:p;\vhole}$ in the case for a
field projection trace.  The added rules $\textsc{SHole}^*$ and
$\textsc{SDiamond}^*$ handle the possibility that a partial pattern reduces to
$\vhole$ or $\vany$ through projection; we did not need to handle this
case earlier because a simple set pattern is either a hole (which
could be handled by the rule $\vhole,T \trbslice [],\tremp$) or a
complete set pattern showing all of the labels of the result.

\begin{figure}[tb]
\fbox{$p, T \trbslice \rho, S$}
\vspace{-4ex}
\begin{smathpar}
\inferrule*[right=SRec]{
\vrec{A{:}p;\vhole},T \trbslice \rho,S
}{
p,\trfield{T}{A} \trbslice \rho,\trfield{S}{A}
}
\and
\inferrule*[right=SProj$_A$]
{
p.A_1,T_1 \trbslice \rho_1,S_1\\
\cdots\\
p.A_n,T_n \trbslice \rho_n,S_n
}
{
p, \trrec{A_1{:}T_1,\ldots,A_n{:}T_n}
\trbslice \rho_1 \sqcup \cdots \sqcup \rho_n,
\kwrec{A_1{:}S_1,\ldots,A_n{:}S_n}
}
\end{smathpar}
\fbox{$p, x.\Theta \trbslice^* \rho, \Theta',p_0$}
\vspace{-1ex}
\begin{smathpar}
\inferrule*[right=SHole$^*$]
{
\strut
}
{
\vhole, x.\Theta \trbslice^* [], \emptyset, \vhole
}
\and
\inferrule*[right=SDiamond$^*$]
{
\strut
}
{
\vany, x.\emptyset \trbslice^* [], \emptyset, \emptyset
}
\end{smathpar}
\caption{Backward trace slicing over enriched patterns.}
\label{fig:nrc-back-slice-extended}
\end{figure}

Both simple and extended patterns satisfy a number of lemmas that are
required to prove the correctness of trace slicing.  
\begin{lemma}[Properties of union and restriction]
~
\begin{enumerate}
\item If $p_1 \sqleq v_1$ and $p_2 \sqleq v_2$ and $v_1 \simat{p_1}
    v_1'$ and $v_2 \simat{p_2} v_2'$ then $v_1 \semu v_2 \simat{p_1
      \semu p_2} v_1' \semu v_2'$, provided all of these disjoint unions
    are defined.
\item If $p \sqleq v_1 \semu v_2$ and $L_1 \leq \dom(v_1)$ and
    $L_2 \leq \dom(v_2)$ and $L_1,L_2$ are prefix-disjoint, then $p|_{L_1} \sqleq v_1$ and
    $p|_{L_2} \sqleq v_2$.
\end{enumerate}
\end{lemma}
\begin{lemma}[Projection and $\sqleq$]
~
  \begin{enumerate}
  \item If $p \sqleq \set{\epsilon.v}$ then $p.\epsilon \sqleq v$.
  \item If $p \sqleq 1\cdot v_1 \semu 2 \cdot v_2$ then $p[1] \sqleq
    v_1$ and $p[2] \sqleq v_2$.
  \item If $p \sqleq \lbl\cdot v $ then $p[\lbl] \sqleq v$.
  \item If $p \sqleq \vrec{\overline{A_i:v_i}}$ then $p.A_i \sqleq
    v_i$.
 \end{enumerate}
\end{lemma}
\begin{lemma}[Projection and $\simat{p}$]
~
 \begin{enumerate}
  \item If $p \sqleq \{\epsilon.v\}$ and $v \simat{p.\epsilon} v'$ then $\set{\epsilon.v} \simat{p}
    \set{\epsilon.v'}$.  
  \item If $p \sqleq 1 \cdot v_1 \semu 2 \cdot v_2$ and $v_1 \simat{p[1]}
    v_1'$ and $v_2 \simat{p[2]} v_2'$ then $1 \cdot v_1 \semu 2 \cdot
    v_2 \simat{p} 1 \cdot v_1' \semu 2 \cdot v_2'$.
  \item If $p \sqleq \lbl \cdot v$ and $v \simat{p[\lbl]} v'$ then
    $\lbl\cdot v \simat{p} \lbl \cdot v'$.
  \item If $p \sqleq \vrec{\overline{A_i:v_i}}$ and $v_1\simat{p.A_1}
    v_1', \ldots, v_n \simat{p.A_n} v_n'$ then $\vrec{\overline{A_i:v_i}} \simat{p} \vrec{\overline{A_i:v_i'}}$.
\end{enumerate}
\end{lemma}
\begin{intronly}
Proofs are collected in 
  Appendix~\ref{app:pattern-proofs}.
\end{intronly}

We now state the key correctness property for slicing.  Intuitively,
it says that if we slice $T$ with respect to output pattern $p$, obtaining
a slice $\rho$ and $S$, then $p$ will be reproduced on recomputation 
under any change to the input and trace that is consistent with the slice ---
formally, that means that the changed trace $T'$ must match the sliced
trace $S$, and the changed input $\gamma'$ must match $\gamma$ modulo $\rho$.

\begin{intronly}
  \begin{theorem}[Correctness of Slicing]~
    \label{thm:trace-slicing}
    \begin{enumerate}
    \item Suppose $\gamma,T \trrun v$ and $p \sqleq v$ and $p, T
      \trbslice \rho, S$.  Then for all $\gamma' \simat{\rho} \gamma$
      and $T' \sqgeq S$ such that $\gamma',T' \trrun v'$ we have $v'
      \simat{p} v$.

    \item Suppose $\gamma,x\in v_0, \Theta \trrun^* v$ and $p \sqleq
      v$ and $p, x.\Theta_0 \trbslice \rho, \Theta_0', p_0$, where
      $\Theta_0 \subseteq \Theta$.  Then for all $\gamma'
      \simat{\rho}\gamma$ and $v_0' \simat{p_0} v_0$ and
      $\Theta'\sqgeq \Theta_0'$ such that $\gamma',x\in v_0,\Theta'
      \trrun^* v'$ we have $v' \simat{p} v$.
    \end{enumerate}
  \end{theorem}
\end{intronly}
\begin{insubonly}
   \begin{theorem}[Correctness of Slicing]~
    \label{thm:trace-slicing}
  Suppose $\gamma,T \trrun v$ and $p \sqleq v$ and $p, T
      \trbslice \rho, S$.  Then for all $\gamma' \simat{\rho} \gamma$
      and $T' \sqgeq S$ such that $\gamma',T' \trrun v'$ we have $v'
      \simat{p} v$.
  \end{theorem}
\end{insubonly}

%The proof is in Appendix \ref{app:proof-trace-slicing}.

Returning to our running example, we can use the enriched pattern $p'
= \{[r_2].\kwrec{B{:}8;\vhole}\}\udot \vhole$ to indicate interest in
the $B$ field of $r_2$, without naming the other fields of the row or
the other row indexes.  Slicing with respect to this pattern yields
the following slice:
\[
T'' =   \trcomp{\_}{x}{R}{\{[r_2].\trifthen{x.B=3}{\_}{\_}{\{\_\}}\}}
\]
and the slice of $R$ is $R'' =
\{[r_2].\kwrec{B{:}3,C{:}8;\vhole}\}\udot \vhole$.  This indicates
that (as before) the values of $r_1$ and $r_3$ and of the $A$ field of
$r_2$ are irrelevant to $r_2$ in the result; unlike $T'$, however, we
can also potentially replay if $r_1$ and $r_3$ have been deleted from
$R$.  Likewise, we can replay if the $A$ field has been removed from a
record, or if some other field such as $D$ is added.  This illustrates
that enriched patterns allow for smaller slices than simple patterns,
with greater flexibility concerning possible updates to the input.

%\todo{Remove this?}
A natural question is whether slicing computes the (or a) smallest
possible answer.  Because our definition of correct slices is based on
recomputation, minimal slices are not computable, by a
straightforward reduction from the undecidability of minimizing
dependency provenance~\cite{cheney11mscs,acar13jcs}.

% LocalWords:  subtraces notationally nrc sim lub lbl ldots emph emph eqnarray

%%% Local Variables: 
%%% mode: latex
%%% TeX-master: "main"
%%% End: 
% LocalWords:  vhole vany rpat vrec vrec overline emptyset vset udot sqcup

%\input{recomputation}
\section{Query and Differential Slicing}
\label{sec:qdslicing}

We now adapt slicing techniques to provide explanations in terms of
query expressions, and show how to use differences between slices to
provide precise explanations for parts of the output.

\subsection{Query slicing}

Our previous work~\cite{perera12icfp} gave an algorithm for extracting
a \emph{program slice} from a trace.  We now adapt this idea to
queries. A trace slice shows the parts of the trace that need to be
replayed in order to compute the desired part of the output;
similarly, a \emph{query slice} shows the part of the query expression
that is needed in order to ensure that the desired part of the output
is recomputed.  As with traces, we allow holes $\vhole$ in programs to
allow deleting subexpressions, and 
% %
% \[e ::= \cdots \mid \vhole\]
% %
define $\sqleq$ as a syntactic precongruence such
that $\vhole \sqleq e$.  We also define a least upper bound operation
$e \sqcup e'$ on partial query expressions in the obvious way, so that
$\hole \sqcup e = e$.

We define a judgment $p,T \uneval \rho,e$ that traverses $T$ and ``unevaluates'' $p$, yielding a partial input environment $\rho$ and partial program $e$.  The rules are illustrated in \figref{nrc-uneval}.  Many of the rules are similar to those for trace slicing; the main differences arise in the cases for conditionals and comprehensions, where we collapse the sliced expressions back into expressions, possibly inserting holes or merging sliced expressions obtained by running the same code in different ways (as in a comprehension that contains a conditional).

\begin{figure}[tb]
  
\fbox{$p, T \uneval \rho, e$}
\begin{smathpar}
\inferrule*
{\strut}
{
\vhole,T \uneval [],\vhole
}
\begin{intronly}
\and
\inferrule*
{
\strut
}
{
p,\trc \uneval [],\trc
}
\and
\inferrule*
{
\vany, T_1 \uneval \rho_1,e_1\\
\cdots\\
\vany,T_n \uneval\rho_n, e_n
}
{
p,\kwf(T_1,\ldots,T_n) \uneval \rho_1 \sqcup \cdots \sqcup
\rho_n, \trf(e_1,\ldots,e_n)
}
\and
\inferrule*
{\strut
}
{
p,x \uneval[x\mapsto p], x
}
\and
\inferrule*
{
p_2,T_2 \uneval \rho_2[x\mapsto p_1],e_2\\
p_1,T_1 \uneval \rho_1, e_1
}
{
p_2,\trlet{x}{T_1}{T_2} \uneval \rho_1 \sqcup \rho_2, \trlet{x}{e_1}{e_2}
}
\and
\inferrule*
{
p_1,T_1 \uneval \rho_1,e_1\\
\cdots\\
p_n,T_n \uneval \rho_n,e_n
}
{
\kwrec{A_1{:}p_1,\ldots,A_n{:}p_n} , \trrec{A_1{:}T_1,\ldots,A_n{:}T_n}
\uneval \rho_1 \sqcup \cdots \sqcup \rho_n,
\kwrec{A_1{:}e_1,\ldots,A_n{:}e_n}
}
\and
\inferrule*{
p.A,T \uneval \rho,S
}{
p,\trfield{T}{A} \uneval \rho,\trfield{S}{A}
}
\end{intronly}
\and
\inferrule*
{
p_1,T_1 \uneval \rho_1,e_1' \\
\kwtrue,T \uneval \rho,e'
}
{
p_1,\trifthen{T}{e_1}{e_2}{T_1} \uneval \rho_1\sqcup \rho, \kwif{e'}{e_1'}{\vhole}
}
\and
\inferrule*
{
p_2,T_2 \uneval \rho_2,e_2' \\
\kwfalse,T \uneval \rho,e'
}
{
p_2,\trifelse{T}{e_1}{e_2}{T_2} \uneval \rho_2\sqcup \rho, \kwif{e'}{\vhole}{e_2'}
}
\begin{intronly}
\and
\inferrule*
{
\strut
}
{
\emptyset, \emptyset \uneval [],\emptyset
}
\and
\inferrule*
{
p[\epsilon], T \uneval \rho,e
}
{
p, \trsets{T} \uneval \rho,\trsets{e}
}
\and
\inferrule*
{
p[1], T_1 \uneval \rho_1, e_1 \\
p[2], T_2 \uneval \rho_2, e_2 
}
{
 p, \trsetu{T_1}{T_2} \uneval \rho_1 \sqcup \rho_2 , \kwsetu{e_1}{e_2}
}
\and
\inferrule*
{
 \vany, T \uneval\rho, e
}
{
 p, \trsetsum{T} \uneval\rho, \trsetsum{e}
}
\and
\inferrule*
{
 \vany, T \uneval \rho,e
}
{
 p, \trsetise{T} \uneval \rho,\trsetise{e}
}
\end{intronly}
\and
\inferrule*{
p, x.\Theta \uneval^* \rho', e', p_0\\
p_0, T \uneval \rho,e_0'
}{
p,\trcomp{e}{x}{T}{\Theta} \uneval \rho \sqcup \rho', \kwcomp{e'}{x}{e_0'}
}
\begin{intronly}
\and
\inferrule*{
\vany,T_1 \uneval \rho_1,e_1\\
\cdots\\
\vany,T_n \uneval \rho_n,e_n
}{
\vany, \trrec{A_1:T_1,\ldots,A_n:T_n}\uneval \rho_1 \sqcup \cdots
\sqcup \rho_n, \trrec{A_1:e_1,\ldots,A_n:e_n}
}
\and
\inferrule*{
\strut
}{
\vany,\emptyset \uneval [],\emptyset
}
\end{intronly}
\end{smathpar}
\fbox{$p, x.\Theta \uneval \rho, e,p'$}
\vspace{-1ex}
\begin{smathpar}
\inferrule*
{
\strut
}
{
\vhole,x. \Theta \uneval^* [], \vhole, \vhole
}
\and
\inferrule*
{
\strut
}
{
\emptyset,x. \emptyset \uneval^* [], \vhole, \emptyset
}
\and
\inferrule*
{
\strut
}
{
\vany,x. \emptyset \uneval^* [], \vhole, \emptyset
}
\and
\inferrule*
{
p[\lbl], T \uneval \rho[x\mapsto p_0],e
}
{
p, x.\{\lbl.T\} \uneval^* \rho, e, \{\lbl.p_0\}
}
\and
\inferrule*
{
p|_{\dom(\Theta_1)}, x.\Theta_1 \uneval^* \rho_1,e_1,p_1\\
p|_{\dom(\Theta_2)},x. \Theta_2 \uneval^* \rho_2,e_2,p_2
}
{
p, x.\Theta_1 \semu \Theta_2 \uneval^* \rho_1 \sqcup \rho_2,
e_1 \sqcup e_2, p_1 \semu p_2
}
\end{smathpar}
  \caption{Unevaluation\subonly{ (selected rules)}.}
  \label{fig:nrc-uneval}
\end{figure}

Again, it is straightforward to show that if $\gamma,e \red v,T$ and
$p \sqleq v$ then there exist $\rho,e'$ such that $p,T \uneval
\rho,e'$, where $\rho \sqleq \gamma$ and $e' \sqleq e$.  The essential
correctness property for query slices is similar to that for trace
slices: again, we require that rerunning any sufficiently similar
query on a sufficiently similar input produces a result that matches
$p$.  
\begin{intronly}
  The proof of this result is in  Appendix~\ref{app:proof-query-slicing}.
\end{intronly}

\begin{intronly}
  \begin{theorem}[Correctness of Query
    Slicing]\label{thm:query-slicing}
    ~
    \begin{enumerate}
    \item Suppose $\gamma,T \trrun v$ and $p \sqleq v$ and $p, T
      \uneval \rho, e$.  Then for all $\gamma' \simat{\rho} \gamma$
      and $e'\sqgeq e$ such that $\gamma',e' \red v'$ we have $v'
      \simat{p} v$.

    \item Suppose $\gamma,x\in v_0, \Theta \trrun^* v$ and $p \sqleq
      v$ and $p, \Theta \uneval \rho, e_0, p_0$.  Then for all
      $\gamma' \simat{\rho}\gamma$ and $v_0' \simat{p_0} v_0$ and
      $e_0'\sqgeq e_0$ such that $\gamma',x\in v_0', e_0' \red^* v'$
      we have $v' \simat{p} v$.
    \end{enumerate}
  \end{theorem}
\end{intronly}
\begin{insubonly}
   \begin{theorem}[Correctness of Query
    Slicing]\label{thm:query-slicing}
        Suppose $\gamma,T \trrun v$ and $p \sqleq v$ and $p, T
      \uneval \rho, e$.  Then for all $\gamma' \simat{\rho} \gamma$
      and $e'\sqgeq e$ such that $\gamma',e' \red v'$ we have $v'
      \simat{p} v$.
  \end{theorem}
\end{insubonly}
Combining with the consistency property (Proposition~\ref{prop:consistency}), we have:
\begin{corollary}
      Suppose $\gamma,e \red v,T$ and $p \sqleq v$ and $p, T \uneval \rho, e'$.  Then
      for all $\gamma' \simat{\rho} \gamma$ and $e'' \sqgeq e'$ such that $\gamma', e'' \red v'$,
    we have $v \eqat{p} v'$.  
\end{corollary}

Continuing our running example, the query slice for the pattern $p'$
considered above is 
\[Q' = \kwcomp{\kwif{x.B=3}{\{\kwrec{A{:}\vhole,B{:}x.C}\}}{\{\}}}{
  x}{R}\] 
since only the computation of the $A$ field in the output is irrelevant
to $p'$.

\subsection{Differential slicing}

We consider a \emph{pattern difference} to be a pair of patterns
$(p,p')$ where $p \sqleq p'$.  Intuitively, a pattern difference
selects the part of a value present in the outer component $p'$ and
not in the inner component $p$.  For example, the pattern difference
$(\kwrec{B{:}\vhole;\vhole},\kwrec{B{:}8;\vhole})$
selects the value $8$ located in the $B$ component of a record.  We
can also write this difference as
$\kwrec{B{:}\fbox{$8$};\vhole}$, using
$\fbox{?}$ to highlight the boundary between the inner and outer
pattern.  Trace and query pattern differences are defined analogously.

It is straightforward to show by induction that slicing is monotonic in both arguments:
\begin{lemma}[Monotonicity]\label{lem:monotonicity}
  If $p \sqleq p'$ and $T \sqleq T'$ and $p',T' \trbslice \rho',S'$
  then there exist $\rho,S$ such that $p,T \trbslice \rho,S$ and $\rho \sqleq \rho'$ and $S \sqleq S'$.  In addition, $p,S' \trbslice \rho,S$.
\end{lemma}
This implies that given a pattern difference and a trace, we can compute a trace difference using the following rule:
\[\inferrule*{
p_2,T \trbslice \rho_2,S_2\\
p_1,S_2 \trbslice \rho_1, S_1
}{
(p_1,p_2),T \trbslice (\rho_1,\rho_2), (S_1,S_2)
} 
\]
It follows from monotonicity that $\rho_1 \sqleq \rho_2$ and $S_1
\sqleq S_2$, thus, pattern differences yield trace differences.
Furthermore, the second part of monotonicity implies that we can
compute the smaller slice $S_1$ from the larger slice $S_2$, rather
than re-traverse $T$. It is also possible to define a simultaneous
differential slicing judgment, as an optimization to ensure we only
traverse the trace once.  

Query slicing is also monotone, so
differential query slices can be obtained in exactly the same way.
Revisiting our running example one last time, consider the
differential pattern $\{[r_2].\kwrec{B{:}\fbox{$8$};\vhole}\}\udot
\vhole$.  The differential query slice for the pattern $p'$
considered above is 
\[Q'' = \kwcomp{\kwif{x.B=3}{\{\kwrec{A{:}\vhole,B{:}\fbox{$x.C$}}\}}{\{\}}}{
  x}{R}\] 
%

% LocalWords:  unevaluates nrc uneval monotonicity hoc yz subterm

%%% Local Variables: 
%%% mode: latex
%%% TeX-master: "main"
%%% End: 

\section{Examples and Discussion}
\label{sec:examples}

In this section we present some more complex examples to  illustrate key points.

\paragraph*{Renaming}
Recall the swapping query from the introduction, written in NRC as 
\[Q_1 = \kwcomp{\{\kwif{x.A > x.B}{\kwrec{A{:}x.B,B{:}x.A}}{\kwrec{x}}\}}{x}{R}\;\]
This query illustrates a key difference between our approach and the
how-provenance model of Green et al.~\cite{DBLP:conf/pods/2007/GreenKT07}.  As discussed in
\cite{cheney09ftdb}, renaming operations are ignored by
how-provenance, so the how-provenance annotations of the results of
$Q_1$ are the same as for a query that simply returns $R$.  In other
words, the choice to swap the fields when $A > B$ is not reflected in
the how-provenance, which shows that it is impossible to extract
where-provenance (or traces) from how-provenance.  Extracting
where-provenance from traces appears straightforward, extending our
previous work~\cite{acar13jcs}.

This example also illustrates how traces and slices can be used for
partial recomputation.  The slice for output pattern $\{[1,r_1].\kwrec{B{:}2}\}
\udot \vhole$, for example, will show that this record was produced
because the $A$ component of $\kwrec{A{:}1,B{:}2,C{:}7}$ at index $[r_1]$ in the
input was less than or equal to the  $B$ component.  Thus, we can
replay after any change that preserves this ordering information.

\paragraph*{Union}
Consider query 
\[Q_2 = \bigcup \{\{\kwrec{B{:}x.B}\} \mid x \in R\} \cup
\{\kwrec{B{:}3}\}\]
that projects the $B$ fields of elements of $R$ and adds another copy
of $\kwrec{B{:}3}$ to the result.  This yields
\[Q_2(R) = \{[1,r_1].\kwrec{B{:}2},[1,r_2].\kwrec{B{:}3},[1,r_3].\kwrec{B{:}3},[2].\kwrec{B{:}3}\}\]
This illustrates that the indexes may not all have the same length,
but still form a prefix code.  If we slice with respect to
$\{[1,r_2].\kwrec{B{:}3}\}\udot \vhole$ then the query slice is:
\[Q_2' = \bigcup \{\{\kwrec{B{:}x.B}\} \mid x \in R\} \cup \vhole\]
and $R' = \{[r_2].\kwrec{B{:}3;\vhole}\} \udot \vhole$
whereas if we slice with respect to $\{[2].\kwrec{B{:}3}\}\udot
\vhole$ then the query slice is $Q_2'' = \vhole \cup \{\kwrec{B{:}3}\}$
and $R'' = \vhole$, indicating that this part of the result has no dependence on the input.

A related point: one may wonder whether it makes sense to select a
particular copy of $\kwrec{B{:}3}$ in the output, since in a
conventional multiset, multiple copies of the same value are
indistinguishable.  We believe it is important to be able  to distinguish
different copies of a value, which may have different explanations.  
Treating $n$ copies of a value as a single value with multiplicity $n$
would obscure this distinction and force us to compute the slices of
all of the copies even if only a single explanation is required. This
is why we have chosen to work with indexed sets, rather than pure
multisets.
\paragraph*{Joins}
So far all examples have involved a single table $R$.  Consider
a simple join query 
\[Q_3 = \{\kwrec{A{:}x.A,B{:}y.C} \mid x \in R, y \in S, x.B=y.B\}\]
and consider the following table $S$, and the result $Q_3(R,S)$.
\[S = \begin{array}{c|cc}
  id & B & C\\
\hline
 {}[s_1] & 2 & 4\\
 {}[s_2] & 3 & 4\\
 {}[s_3] & 4 & 5
\end{array}\quad
Q_3(R,S) = \begin{array}{c|cc}
  id & A & B \\
\hline
{}[r_1,s_1] & 1 & 4\\
{}[r_2,s_2] & 2 & 4\\
{}[r_3,s_2] & 4 & 5
\end{array}
\]
The full trace of this query execution is as follows:
\[\small
\begin{array}{ll}
T_3 =   \trcomp{\_}{x}{R}{\{\\
\quad
[r_1].\trcomp{\_}{y}{S}{\{&[s_1].\trifthen{x.B=y.B}{\_}{\_}{\{\_\}},\\
&[s_2].\trifelse{x.B=y.B}{\_}{\_}{\{\}},\\
&[s_3].\trifelse{x.B=y.B}{\_}{\_}{\{\}}\}},\\
\quad
[r_2].\trcomp{\_}{y}{S}{\{&[s_1].\trifelse{x.B=y.B}{\_}{\_}{\{\}},\\
&[s_2].\trifthen{x.B=y.B}{\_}{\_}{\{\_\}},\\
&[s_3].\trifelse{x.B=y.B}{\_}{\_}{\{\}}\}},\\
\quad
[r_3].\trcomp{\_}{y}{S}{\{&[s_1].\trifelse{x.B=y.B}{\_}{\_}{\{\}},\\
&[s_2].\trifthen{x.B=y.B}{\_}{\_}{\{\_\}},\\
&[s_3].\trifelse{x.B=y.B}{\_}{\_}{\{\}}\}}\}}
\end{array}
\]
 Slicing with respect to
$\{[r_1,s_1].\kwrec{A{:}1;\vhole},[r_2,s_2].\kwrec{B{:}4;\vhole}\} \udot \vhole$ yields
trace slice
\[\small
\begin{array}{ll}
T_3' =   \trcomp{\_}{x}{R}{\{\\
\quad
[r_1].\trcomp{\_}{y}{S}{\{&[s_1].\trifthen{x.B=y.B}{\_}{\_}{\{\_\}}},\\
\quad
[r_2].\trcomp{\_}{y}{S}{\{&[s_2].\trifthen{x.B=y.B}{\_}{\_}{\{\_\}}}\}}
\end{array}
\]
and input slice $R_3' =
\{[r_1].\kwrec{A{:}1,B{:}2;\vhole},[r_2].\kwrec{B{:}3;\vhole}\}$ and
$S_3' = \{[s_1].\kwrec{B{:}2;\vhole},[s_2].\kwrec{B{:}3;C{:}4}\}$.

\paragraph*{Workflows}
NRC expressions can be used to represent workflows, if primitive operations are
added representing the workflow steps~\cite{DBLP:conf/dils/HiddersKSTB07,acar10tapp}.
To illustrate query slicing and differential slicing for a
workflow-style query, consider the following more complex query:
\[Q_4 = \{f(x,y) \mid x \in T, y \in T, z \in U, p(x,y),q(x,y,z)\}\]
where $f$ computes some function of $x$ and $y$ and $p$ and $q$ are
selection criteria.  Here, we assume $T$ and $U$ are collections of data files
and $f,p,q$ are additional primitive operations on them.
 This query exercises most of the distinctive features of our
approach; we can of course translate it to the NRC core calculus used
in the rest of the paper.  If $dom(T) = \{t_1,\ldots,t_{10}\}$ and $dom(U) =
\{u_1,\ldots,u_{10}\}$ then we might obtain
result $\{[t_3,t_4,u_5].v_1, [t_6,t_8,u_{10}].v_2\}$.  If
we focus on the value $v_1$ using the pattern $\{[t_3,t_4,u_5].v_1\}\udot
\vhole$, then the program slice we obtain is $Q_4$ itself,
while the data slice might be $T' = \{[t_3]. \vany,[t_4]. \vany\}
\udot \vhole, U' = \{[u_5].\vany\} \udot
\vhole$, indicating that if the values at $t_3,t_4,u_5$ are held
fixed then the end result will still be $v_1$.  The trace slice is
similar, and shows that $v_1$ was computed by applying $f$ with $x$
bound to the value at $t_3$ in $T$, $y$ bound to $t_4$, and $z$ bound
to $u_5$, and that $p(x,y)$ and $q(x,y,z)$ succeeded for these values.

If we consider a differential slice using pattern difference 
$\{[t_3,t_4,u_5].\fbox{$v_1$}\}\udot \vhole$
 then we obtain the following program difference:
\[Q_4^\delta = \bigcup\{\fbox{$f(x,y)$} \mid x \in T, y \in T, z \in U, p(x,y),q(x,y,z)\}
\]
This shows that most of the query is needed to
ensure that the result at $[t_3,t_4,u_5]$ is produced, but the subterm $f(x,y)$
is only needed to compute the value $v_1$.  This can be viewed as a
query-based explanation for this part of the result.

%%% Local Variables: 
%%% mode: latex
%%% TeX-master: "main"
%%% End: 

\section{Implementation}
\label{sec:implementation}

To validate our design and experiment with larger examples, we
extended our Haskell implementation \Slicer of program slicing for
functional programs~\cite{perera12icfp} with the traces and slicing
techniques presented in this paper.  We call the resulting system
\NRCSlicer; it supports a free combination of NRC and general-purpose
functional programming features.  \NRCSlicer interprets expressions
in-memory without optimization.  As
reported previously for \Slicer, we have experimented with several
alternative tracing and slicing strategies, which use Haskell's lazy
evaluation strategy in different ways.  The alternatives we consider
here are:
\begin{itemize}
\item \emph{eager}: the trace is fully computed during evaluation.
\item \emph{lazy}: the value is computed eagerly, but the trace is
  computed lazily using Haskell's default lazy evaluation strategy.
% \item \emph{delayed}: the trace is computed lazily, and only up to a
%   certain bound; during slicing, if a delayed subtrace operation is
%   encountered, that part of the computation is restarted with tracing enabled.
\end{itemize}
To evaluate  the effectiveness of  enriched patterns, we  measured the
time needed for  the eager and lazy techniques to  trace and slice the
workflow example $Q_4$  in   the  previous  section.  We considered a
instantiation of the workflow where the data values are
simply integers and with   input  tables $T,U =
\{1,\ldots,50\}$, and defined the operations $f(x,y)$ as $x*y$,
$p(x,y)$ as $x < y$, and $q(x,y,z)$ as $x^2+y^2=z^2$.   This
is not a realistic workflow, and we expect that the time to evaluate
the basic operations of a
realistic workflow following this pattern would be much larger.  However,
the overheads of tracing and slicing do not depend on the
execution time of primitive operations, so we can still draw some
conclusions from this simplistic example.

The comprehension iterates
over $50^3$ = 125,000 triples, producing 20 results.  We considered
simple and enriched patterns selecting a single element of the
result.  We measured evaluation time, and the overhead of tracing,
trace slicing, and query slicing.  The  experiments  were
conducted on a  MacBook Pro with 2GB RAM and a  2.8GHz Intel Core Duo,
using GHC version 7.4.
\begin{center}
  \begin{tabular}{|l|cccc|}
\hline
          & eval &  trace & slice & qslice\\
\hline
eager-simple & 0.5 & 1.5  & 2.5 &  1.6 \\
eager-enriched & 0.5 & 1.5 & $<$0.1 & $<$0.1 \\
lazy-simple & 0.5 & 0.7 & 1.3& 1.7\\
lazy-enriched & 0.5 & 0.7 & $<$0.1& $<$0.1\\
\hline
 \end{tabular}
\end{center}
The times are in seconds. The ``eval'' column shows the time needed to
compute the result without tracing.  The ``trace'', ``slice'', and
``qslice'' columns show the added time needed to trace and compute
slices.  The full traces in each of these runs have over 2.1 million
nodes; the simple pattern slices are almost as large, while the
enriched pattern slices are only 95 nodes.  For this example, slicing
is over an order of magnitude faster using enriched patterns.  The
lazy tracing approach required less total time both for tracing and
slicing (particularly for simple patterns).  Thus, Haskell's built-in
lazy evaluation strategy offers advantages by avoiding explicitly
constructing the full trace in memory when it is not needed; however,
there is still room for improvement.  Again, however, for an actual
workflow involving images or large data files, the evaluation time
would be much larger, dwarfing the time for tracing or slicing.

Our implementation is a proof-of-concept that evaluates queries
in-memory via interpretation, rather than compilation; further work
would be needed to adapt our approach to support fine-grained
provenance for conventional database systems.  Nevertheless, our
experimental results do suggest that the lazy tracing strategy and use
of enriched patterns can effectively decrease the overhead of tracing,
making it feasible for in-memory execution of workflows represented in
NRC.

% LocalWords:  Haskell Haskell's MacBook GHC cccc eval qslice

%%% Local Variables: 
%%% mode: latex
%%% TeX-master: "main"
%%% End: 

\section{Related and future work}
\label{sec:related}

Program slicing has been studied
extensively~\cite{weiser81icse,DBLP:journals/jpl/Tip95,field98ist}, as
has the use of execution traces, for example in dynamic slicing.  Our
work contrasts with much of this work in that we regard the trace and
underlying data as being of interest, not just the program.
Some of our previous work~\cite{cheney11mscs} identified analogies
between program slicing and provenance, but to our knowledge, there is
no other prior work on slicing in databases.

%%% Previous work - maybe overlaps with related work, summarize

% As mentioned in the introduction, threads of research differ in several respects:
% they allow different types of provenance queries, and provide
% different levels of granularity.  Work on database provenance has
% typically been based on stronger formal guarantees, building on the
% mature theory of relational query languages, but while this work has been
% implemented in a number of research prototypes (such
% as~\cite{karvounarakis10sigmod,amsterdamer11pvldb}, among others), it
% has not yet been widely adopted by database systems. Meanwhile,
% workflow provenance models have been implemented for many workflow
% systems, but they have typically been system-specific, with
% relatively little generalization or foundational work aside from work on common
% data models for exchanging their provenance data, such as the Open
% Provenance Model (OPM) or W3C's
% PROV~\cite{opm,prov}.

%%% End previous work
 
% As discussed in the introduction, there is a great deal of research on
% developing systems with rich provenance-tracking features.  Several
% surveys~\cite{bose05cs,cheney09ftdb,DBLP:journals/sigmod/SimmhanPG05}
% offer an overview of the area.  We will restrict discussion to 
% foundational work.

Lineage and why-provenance were motivated semantically in terms of
identifying \emph{witnesses}, or parts of the input needed to ensure
that a given part of the output is produced by a query.  Early work on
lineage in relational algebra~\cite{DBLP:journals/tods/CuiWW00}
associates each output record with a witness.  Buneman et al. studied
a more general notion called why-provenance that maps an output part
to a collection of witnesses~\cite{buneman01icdt,buneman02pods}.  This
idea was generalized further to the \emph{how-provenance} or
\emph{semiring}
model~\cite{DBLP:conf/pods/2007/GreenKT07,DBLP:conf/pods/FosterGT08},
based on using algebraic expressions as annotations; this approach has
been extended to handle some forms of negation and
aggregation~\cite{amsterdamer11pods,geerts10jal}.  Semiring
homomorphisms commute with query evaluation; thus, homomorphic changes
to the input can be performed directly on the output without
re-running the query.  However, this approach only applies to changes
describable as semiring homomorphisms, such as deletion.

Where-provenance was also introduced by Buneman et
al.~\cite{buneman01icdt,buneman02pods}.  Although the idea of tracking
where input data was copied from is natural, it is nontrivial to
characterize semantically, because where-provenance does not always
respect semantic equivalence.  In later work, Buneman et
al.~\cite{buneman08tods} studied where-provenance for the pure NRC and
characterized its expressiveness for queries and updates.  It would be
interesting to see whether their
notion of \emph{expressive completeness}  for where-provenance could
be extended to richer provenance models, such as  traces, possibly
leading to an implementation strategy via translation to plain NRC.
% %
% \uremar{The last sentence makes it sounds like equality a type.  I
%   agree that we should say something about how they are not general.
%   Isn't this as easy as saying beyond straightline code?  doesn't
%   where provenance break as soon as you throw conditionals in?}

Provenance has been studied extensively for scientific workflow
systems~\cite{bose05cs,DBLP:journals/sigmod/SimmhanPG05}, but there
has been little formal work on the semantics of workflow provenance.
The closest work to ours is that of Hidders et
al.~\cite{DBLP:conf/dils/HiddersKSTB07}, who model workflows by
extending the NRC with nondeterministic, external function calls.
They sketch an operational semantics that records \emph{runs} that
contain essentially all of the information in a derivation tree,
represented as a set of triples.  They also suggest ways of extracting
\emph{subruns} from runs, but their treatment is
partial and lacks strong formal guarantees analogous to our
results.  

There have been some attempts to reconcile the database and workflow views of
provenance; Hidders et al.~\cite{DBLP:conf/dils/HiddersKSTB07} argued
for the use of Nested Relational Calculus (NRC) as a unifying
formalism for both workflow and database operations, and subsequently
Kwasnikowska and Van den Bussche~\cite{kwasnikowska08ipaw} showed how
to map this model to the Open Provenance Model.  Acar et
al.~\cite{acar10tapp} later formalized a graph model of provenance for
NRC.  The most advanced work in this direction appears to be that of
Amsterdamer et al.~\cite{amsterdamer11pvldb}, who combined workflow
and database styles of provenance in the context of the PigLatin
system (a MapReduce variant based on nested relational queries).
Lipstick allows analyzing the impact of
restricted hypothetical changes (such as deletion) on parts of the
output, but to our knowledge no previous work provides a formal
guarantee about the impact of changes other than deletion.
% Other systems, such as the ``Golden-Trail'' and
% also allow for extracting parts of provenance graphs, and the Lipstick
% system of Amsterdamer et al.~\cite{amsterdamer11pvldb} allows
% analyzing the impact of restricted hypothetical changes (such as
% deletion) on parts of the output, but to our knowledge no previous
% work provides a formal guarantee about the impact of a larger classes
% of changes.

In our previous work~\cite{cheney11mscs}, we introduced
\emph{dependency provenance}, which conservatively over-approximates
the changes that can take place in the output if the input is changed.
We developed definitions and techniques for dependency provenance in
full NRC including nonmonotone operations ($\kw{empty}$, $\kw{sum}$) and
primitive functions.
% We view the techniques in this paper as natural generalizations of
%dependency provenance.
Dependency provenance cannot predict exactly how the output will be
affected by a general modification to the source, but it can guarantee
that some parts of the output will not change if certain parts of the
input are fixed.  Our notion of equivalence modulo a pattern is a
generalization of the \emph{equal-except-at} relation used in that
work.
%
%Dependency provenance, however, motivated the techniques in this paper
%by showing that it is possible to approximate completeness under
%general modifications.
%
%Traces are more descriptive because they record not only which parts
%of the input were used but how they were combined to form the output.
%
%
%
%Our earlier technical report~\cite{tr} developed a model of traced
%evaluation for NRC and proved elementary properties corresponding to
%Theorem~\ref{thm:completeness-trace}, but did not investigate slicing.
%Chong~\cite{DBLP:conf/fast/Chong09} used it as a basis for studying
%provenance security.  The current paper supersedes that manuscript:
%our model is simpler and cleaner and our main results are stronger.
%
Motivated by dependency provenance, an earlier technical
report~\cite{traces-tr} presented a model of traced evaluation for NRC
and proved elementary properties such as fidelity.  However, it did
not investigate slicing techniques, and used nondeterministic label
generation instead of our deterministic scheme; our deterministic
approach greatly simplifies several aspects of the system,
particularly for slicing.

There are several intriguing directions for future work, including
developing more efficient techniques for traced evaluation and slicing
that build upon existing database query optimization capabilities.  It
appears possible to translate multiset queries so as to make the labels
explicit, since a fixed given query increases the label depth by at most a
constant. Thus, it may be possible to evaluate queries with label
information but without tracing
first, then gradually build the trace by slicing backwards through the
query, re-evaluating subexpressions as necessary.  Other 
interesting directions include the use of slicing techniques for security,
to hide confidential input information while disclosing enough about
the trace to permit recomputation, and the possibility of extracting
other forms of provenance from traces, as explored in
the context of functional programs in prior work~\cite{acar13jcs}.

% LocalWords:  Buneman et al semiring homomorphic workflow Hidders workflows
% LocalWords:  nondeterministic subruns nonmonotone Chong Amsterdamer mutliset
% LocalWords:  subexpressions recomputation multiset

%%% Local Variables: 
%%% mode: latex
%%% TeX-master: "main"
%%% End: 

\section{Conclusion}\label{sec:concl}

The importance of provenance for transparency and reproducibility is
widely recognized, yet there has been little explicit discussion of
correctness properties formalizing intuitions about how provenance is
to provide reproducibility. In
self-explaining computation, traces are considered to be explanations
of a computation in the sense that the trace can be used to recompute
(parts of) the output under hypothetical changes to the input.  This
paper develops the foundations of self-explaining computation for
database queries, by defining a tracing semantics for NRC, proposing a
formal definition of correctness for tracing (fidelity) and slicing,
and defining a correct (though potentially overapproximate) algorithm
for trace slicing.  Trace slicing can be used to obtain smaller
``golden trail'' traces that explain only a part of the input or
output, and explore the impact of changes in hypothetical scenarios
similar to the original run.  At a technical level, the main
contributions are the careful use of prefix codes to label multiset
elements, and the development of enriched patterns that allow more
precise slices.  Our design is validated by a proof-of-concept
implementation that shows that laziness and enriched patterns can
significantly improve performance for small (in-memory) examples.

In the near term, we plan to combine our work on self-explaining
functional programs~\cite{perera12icfp} and database queries (this
paper) to obtain slicing and provenance models for programming
languages with query primitives, such as F\#~\cite{cheney13icfp} or
Links~\cite{lindley12tldi}.  Ultimately, our aim is to extend
self-explaining computation to programs that combine several execution
models, including workflows, databases, conventional programming
languages, Web interaction, or cloud computing.

\paragraph{Acknowledgments}
We are grateful to Peter Buneman, Jan Van den Bussche, and Roly Perera
for comments on this work and to the anonymous reviewers for detailed
suggestions.
Effort sponsored by the Air Force Office of Scientific Research, Air
Force Material Command, USAF, under grant number FA8655-13-1-3006. The
U.S. Government and University of Edinburgh are authorized to
reproduce and distribute reprints for their purposes notwithstanding
any copyright notation thereon.  Cheney is supported by a Royal
Society University Research Fellowship, by the EU FP7 DIACHRON
project, and EPSRC grant EP/K020218/1.  Acar is partially supported by
an EU ERC grant (2012-StG 308246---DeepSea) and an NSF grant
(CCF-1320563).
% LocalWords:  reproducibility recomputation

%%% Local Variables: 
%%% mode: latex
%%% TeX-master: "main"
%%% End: 

{
\small
\bibliographystyle{abbrv}

\bibliography{main}

\begin{thebibliography}{10}

\bibitem{acar13jcs}
U.~A. Acar, A.~Ahmed, J.~Cheney, and R.~Perera.
\newblock A core calculus for provenance.
\newblock {\em Journal of Computer Security}, 21:919--969, 2013.
\newblock Full version of a POST 2012 paper.

\bibitem{acar10tapp}
U.~A. Acar, P.~Buneman, J.~Cheney, N.~Kwasnikowska, J.~{Van~den~Bussche}, and
  S.~Vansummeren.
\newblock A graph model of data and workflow provenance.
\newblock In {\em TAPP}, 2010.

\bibitem{amsterdamer11pvldb}
Y.~Amsterdamer, S.~B. Davidson, D.~Deutch, T.~Milo, J.~Stoyanovich, and
  V.~Tannen.
\newblock Putting lipstick on pig: Enabling database-style workflow provenance.
\newblock {\em PVLDB}, 5(4):346--357, 2011.

\bibitem{amsterdamer11pods}
Y.~Amsterdamer, D.~Deutch, and V.~Tannen.
\newblock Provenance for aggregate queries.
\newblock In {\em PODS}, pages 153--164. ACM, 2011.

\bibitem{bose05cs}
R.~Bose and J.~Frew.
\newblock Lineage retrieval for scientific data processing: a survey.
\newblock {\em ACM Comput. Surv.}, 37(1):1--28, 2005.

\bibitem{buneman08tods}
P.~Buneman, J.~Cheney, and S.~Vansummeren.
\newblock On the expressiveness of implicit provenance in query and update
  languages.
\newblock {\em ACM Transactions on Database Systems}, 33(4):28, November 2008.

\bibitem{buneman01icdt}
P.~Buneman, S.~Khanna, and W.~Tan.
\newblock Why and where: A characterization of data provenance.
\newblock In {\em ICDT}, number 1973 in LNCS, pages 316--330. Springer, 2001.

\bibitem{buneman02pods}
P.~Buneman, S.~Khanna, and W.~Tan.
\newblock On propagation of deletions and annotations through views.
\newblock In {\em {PODS}}, pages 150--158, 2002.

\bibitem{buneman95tcs}
P.~Buneman, S.~A. Naqvi, V.~Tannen, and L.~Wong.
\newblock Principles of programming with complex objects and collection types.
\newblock {\em Theor. Comp. Sci.}, 149(1):3--48, 1995.

\bibitem{traces-tr}
J.~Cheney, U.~A. Acar, and A.~Ahmed.
\newblock Provenance traces.
\newblock {\em CoRR}, arXiv.org/abs/0812.0564, 2008.

\bibitem{cheney13pbf}
J.~Cheney, U.~A. Acar, and R.~Perera.
\newblock Toward a theory of self-explaining computation.
\newblock In {\em In search of elegance in the theory and practice of
  computation: a Festschrift in honour of Peter Buneman}, number 8000 in LNCS,
  pages 193--216. Springer, 2013.

\bibitem{cheney11mscs}
J.~Cheney, A.~Ahmed, and U.~A. Acar.
\newblock Provenance as dependency analysis.
\newblock {\em Mathematical Structures in Computer Science}, 21(6):1301--1337,
  2011.

\bibitem{cheney09ftdb}
J.~Cheney, L.~Chiticariu, and W.~C. Tan.
\newblock Provenance in databases: Why, how, and where.
\newblock {\em Foundations and Trends in Databases}, 1(4):379--474, 2009.

\bibitem{cheney13icfp}
J.~Cheney, S.~Lindley, and P.~Wadler.
\newblock A practical theory of language-integrated query.
\newblock In {\em ICFP}, pages 403--416, New York, NY, USA, 2013. ACM.

\bibitem{DBLP:journals/tods/CuiWW00}
Y.~Cui, J.~Widom, and J.~L. Wiener.
\newblock Tracing the lineage of view data in a warehousing environment.
\newblock {\em ACM Trans. Database Syst.}, 25(2):179--227, 2000.

\bibitem{mapreduce}
J.~Dean and S.~Ghemawat.
\newblock {MapReduce}: simplified data processing on large clusters.
\newblock {\em Commun. ACM}, 51(1):107--113, 2008.

\bibitem{field98ist}
J.~Field and F.~Tip.
\newblock Dynamic dependence in term rewriting systems and its application to
  program slicing.
\newblock {\em Information and Software Technology}, 40(11--12):609--636, 1998.

\bibitem{DBLP:conf/pods/FosterGT08}
J.~N. Foster, T.~J. Green, and V.~Tannen.
\newblock Annotated {XML}: queries and provenance.
\newblock In {\em PODS}, pages 271--280, 2008.

\bibitem{geerts10jal}
F.~Geerts and A.~Poggi.
\newblock On database query languages for {$K$}-relations.
\newblock {\em J. Applied Logic}, 8(2):173--185, 2010.

\bibitem{DBLP:conf/pods/2007/GreenKT07}
T.~J. Green, G.~Karvounarakis, and V.~Tannen.
\newblock Provenance semirings.
\newblock In {\em PODS}, pages 31--40, 2007.

\bibitem{DBLP:conf/dils/HiddersKSTB07}
J.~Hidders, N.~Kwasnikowska, J.~Sroka, J.~Tyszkiewicz, and J.~{Van den
  Bussche}.
\newblock A formal model of dataflow repositories.
\newblock In {\em DILS}, 2007.

\bibitem{kwasnikowska08ipaw}
N.~Kwasnikowska and J.~{Van den Bussche}.
\newblock Mapping the {NRC} dataflow model to the open provenance model.
\newblock In {\em IPAW}, pages 3--16, 2008.

\bibitem{lindley12tldi}
S.~Lindley and J.~Cheney.
\newblock Row-based effect types for database integration.
\newblock In {\em TLDI}, pages 91--102. ACM Press, 2012.

\bibitem{missier11ijdc}
P.~Missier, B.~Lud\"ascher, S.~Dey, M.~Wang, T.~McPhillips, S.~Bowers, M.~Agun,
  and I.~Altintas.
\newblock Golden trail: Retrieving the data history that matters from a
  comprehensive provenance repository.
\newblock {\em International Journal of Digital Curation}, 7(1):139--150, 2011.

\bibitem{moreau10ftws}
L.~Moreau.
\newblock The foundations for provenance on the web.
\newblock {\em Foundations and Trends in Web Science}, 2(2--3), 2010.

\bibitem{piglatin}
C.~Olston, B.~Reed, U.~Srivastava, R.~Kumar, and A.~Tomkins.
\newblock Pig latin: a not-so-foreign language for data processing.
\newblock In {\em SIGMOD}, pages 1099--1110. ACM, 2008.

\bibitem{perera12icfp}
R.~Perera, U.~A. Acar, J.~Cheney, and P.~B. Levy.
\newblock Functional programs that explain their work.
\newblock In {\em ICFP}, pages 365--376. ACM, 2012.

\bibitem{sabelfeld03sac}
A.~Sabelfeld and A.~Myers.
\newblock Language-based information-flow security.
\newblock {\em IEEE Journal on Selected Areas in Communications}, 21(1):5--19,
  2003.

\bibitem{DBLP:journals/sigmod/SimmhanPG05}
Y.~Simmhan, B.~Plale, and D.~Gannon.
\newblock A survey of data provenance in e-science.
\newblock {\em SIGMOD Record}, 34(3):31--36, 2005.

\bibitem{DBLP:journals/jpl/Tip95}
F.~Tip.
\newblock A survey of program slicing techniques.
\newblock {\em J. Prog. Lang.}, 3(3), 1995.

\bibitem{weiser81icse}
M.~Weiser.
\newblock Program slicing.
\newblock In {\em ICSE}, pages 439--449. IEEE Press, 1981.

\end{thebibliography}
}

\newpage
\appendix
\onecolumn
 % \section{Proofs of basic properties}
%  \label{app:basic-proofs}

% % Mostly standard. State induction hypothesis for fidelity.

% We next consider properties of patterns and the operations and
% relations on them; these properties are needed to prove the
% correctness of slicing.

 \section{Auxiliary definitions}
 \label{app:auxiliary}
% This appendix collects the typing rules for traces, the subtrace
% relation, and the full definitions of the least upper bound, hole
% substitution, and pattern equivalence operations on enriched patterns.

\figref{nrc-trace-types} summarizes the typing rules for traces.
\figref{nrc-subtrace} defines the subtrace relation.

 \begin{figure*}[tb]
\fbox{$\Gamma \ts T : \tau$}
\vspace{-1ex}
\begin{smathpar}
\inferrule*{n \in \mathbb{N}}
{\Gamma \ts n : \tyint}
\and
\inferrule*{b \in \{\kwtrue,\kwfalse\}}
{\Gamma \ts b : \tybool}
\and
\inferrule*{\kwf : (b_1,\ldots,b_n) \to b \in \Sigma\\
\Gamma \ts T_1 : b_1\\
\cdots\\
\Gamma \ts T_n : b_n
}{
\Gamma \ts \kwf(T_1,\ldots,T_n) : b
}
\and
\inferrule*{x:\tau \in \Gamma
}{
\Gamma \ts x : \tau
}
\and
\inferrule*{
\Gamma \ts T_1 : \tau_1 \\
\Gamma,x:\tau_1 \ts T_2 : \tau_2
}{
\Gamma \ts \trlet{x}{T_1}{T_2} : \tau_2
}
\and
\inferrule*{
\Gamma \ts T_1 :\tau_1 \\
\cdots \\
\Gamma \ts T_n : \tau_n
}{
\Gamma \ts \trrec{A_1:T_1,\ldots,A_n:T_n} : \tyrec{A_1 : \tau_1,\ldots,A_n:\tau_n}
}
\and
\inferrule*{
\Gamma \ts T : \tyrec{A_1 : \tau_1,\ldots,A_n:\tau_n}
}{
\Gamma \ts \trfield{T}{A_i} : \tau_i
}
\and
\inferrule*{
b \in \{\kwtrue,\kwfalse\}\\
\Gamma \ts T : \tybool\\
\Gamma \ts e_1 : \tau\\
\Gamma \ts e_2 : \tau\\
\Gamma \ts T': \tau
}{
\Gamma \ts \trif{T}{e_1}{e_2}{b}{T'} : \tau
}
\and
\inferrule*
{\strut}
{\Gamma \ts \emptyset : \set{\tau}}
\and
\inferrule*
{\Gamma \ts T : \tau }
{\Gamma \ts \trsets{T} : \set{\tau}}
\and
\inferrule*
{\Gamma \ts T : \set{\tau} \\
 \Gamma \ts T' : \set{\tau}}
{\Gamma \ts \trsetu{T}{T'} : \set{\tau}}
\and
%
% \inferrule*
% {\Gamma \ts T : \set{\set{\tau}}}
% {\Gamma \ts \trflatten{T} : \set{\tau}}
%
\inferrule*
{\Gamma \ts T : \set{\tyint}}
{\Gamma \ts \trsetsum{T} : \tyint}
\and
\inferrule*
{\Gamma \ts T : \set{\tau}}
{\Gamma \ts \trsetise{T} : \tybool}
\and
\inferrule*
{\Gamma, x:\tau \ts e : \set{\tau'} \\
 \Gamma \ts T : \set{\tau} \\
\Gamma,x\in \set{\tau} \vdash \Theta : \set{\tau'}}
{ 
\Gamma \ts \trcomp{e}{x}{T}{\Theta} : \set{\tau'}
}
\end{smathpar}
\fbox{$\Gamma,x\in\set{\tau} \vdash \Theta : \set{\tau'}$}
\vspace{-3ex}
\begin{smathpar}
\inferrule*{
\strut
}{
\Gamma,x\in\set{\tau} \vdash \emptyset :  \set{\tau'}
}
\and
\inferrule*{
\Gamma,x:\tau \vdash T:\{\tau'\}
}{
\Gamma,x\in\set{\tau} \vdash \{\lbl.T\} : \{\tau'\}
}
\and
\inferrule*{
\Gamma,x\in\set{\tau} \vdash \Theta_1 :\{\tau'\}\\
\Gamma,x\in\set{\tau} \vdash \Theta_2 :\{\tau'\}
}{
\Gamma,x\in\set{\tau} \vdash \Theta_1 \cup \Theta_2 : \{\tau'\}
}
\and
\end{smathpar}
\caption{Well-typed traces.}
\label{fig:nrc-trace-types}
\end{figure*}

\begin{figure*}[tb]
\fbox{$T' \trsub T$}
\vspace{-2ex}
\begin{smathpar}
\inferrule*
{\strut
}
{
\tremp \trsub T
}
\and
\inferrule*
{\strut
}
{
T \trsub T
}
\and
\inferrule*
{T'' \trsub T'\\
T' \trsub T
}
{
T'' \trsub T
}
\and
\inferrule*
{
T_1' \trsub T_1\\
\cdots \\
T_n' \trsub T_n
}
{
\trf(T_1',\ldots,T_n') \trsub \trf(T_1,\ldots,T_n)
}
\and
\inferrule*
{T_1' \trsub T_1\\
T_2' \trsub T_2
}
{
\trlet{x}{T_1'}{T_2'} \trsub \trlet{x}{T_1}{T_2}
}
\and
\inferrule*
{
T_1' \trsub T_1\\
\cdots \\
T_n' \trsub T_n
}
{
\trrec{A_1:T_1',\ldots,A_n:T_n'} \trsub \trrec{A_1:T_1,\ldots,A_n:T_n}
}
\and
\inferrule*
{
T' \trsub T
}
{
\trfield{T'}{A} \trsub \trfield{T}{A}
}
\and
\inferrule*
{
T_0' \trsub T_0\\
T' \trsub T
}
{
\trif{T_0'}{e_1}{e_2}{b}{T'} \trsub \trif{T_0}{e_1}{e_2}{b}{T}
}
%
% \and
% %
% \inferrule*
% {
% T' \trsub T\\
% T_2' \trsub T_2
% }
% {
% \trifelse{T'}{e_1}{e_2}{T_2'} \trsub \trifelse{T}{e_1}{e_2}{T_1}
% }
%
\and
\inferrule*
{\strut}
{\emptyset \trsub \emptyset}
\and
\inferrule*
{T' \trsub T}
{\trsets{T'} \trsub \trsets{T}}
\and
\inferrule*
{T_1' \trsub T_1 \\
 T_2' \trsub T_2}
{\trsetu{T_1'}{T_2'} \trsub \trsetu{T_1}{T_2}}
\and
%
% \inferrule*
% {T' \trsub T}
% {\trflatten{T'} \trsub \trflatten{T}}
% %
% \and
%
\inferrule*
{T' \trsub T}
{\trsetsum{T'} \trsub \trsetsum{T}}
\and
\inferrule*
{T' \trsub T}
{\trsetise{T'} \trsub \trsetise{T}}
\and
\inferrule*
{ T' \trsub T\\
\Theta' \trsub \Theta}
{\trcomp{e}{x}{T'}{\Theta'} \trsub \trcomp{e}{x}{T}{\Theta}}
\end{smathpar}
\[\Theta' \trsub \Theta \iff \forall \lbl\in \dom(\Theta'). \Theta'(\lbl)
\trsub \Theta(\lbl)\]
\caption{Subtrace relation.}
\label{fig:nrc-subtrace}
\end{figure*}

\section{Proofs of pattern properties}
\label{app:pattern-proofs}
We prove the required properties for enriched patterns.  The
corresponding properties for the sublanguage of simple patterns follow
immediately since simple patterns are closed under the relevant
operations.
\begin{lemma}\label{lem:lub-properties}
  If $p,p' \sqleq v$ then $p \sqcup p'$ exists and is the least upper
  bound of $p$ and $p'$.
\end{lemma}
\begin{proof}
  If both $p$ and $p'$ match some pattern $q$, then it is straightforward to show by induction on $q$ that $p \sqcup p'$ is defined and $p \sqcup p' \sqleq q$.   Specifically, if $p$ or $p'$ is $\vhole$ or $\vany$ then we are done; if $p$ and $p'$ are both constants then we are done; otherwise, in each case, the toplevel structure of $p$ and $p'$ must match $q$, so that we can apply one of the rules for $\sqcup$ on smaller terms that match part of $q$.  When $q = v$, the desired result follows.  The second part (that $\sqcup$ is a least upper bound) also follows directly since clearly, $p,p' \sqleq p \sqcup p'$ and if $p,p' \sqleq q$ then a similar argument shows that $p\sqcup p' \sqleq q$.
\end{proof}
\begin{lemma}\label{lem:sub-sim}
  For any $v,v',p$,  $v\simat{p[\vany/\vhole]} v'$ holds if and only
  if $v \simat{p} v'$  holds and $v= v'$.
\end{lemma}
\begin{proof}
  Straightforward induction on the derivation of $v
  \simat{p[\vany/\vhole]} v'$.
\end{proof}
\begin{lemma}\label{lem:lub-sim}
  For any $v,v',p,p'$, we have $v \simat{p\sqcup p'} v'$ if and only
  if $v \simat{p} v'$ and $v \simat{p'} v'$.  Moreover,  $p \sqleq p'$
  if and only if for all $v,v'$, we have $v \simat{p'} v'$ implies $v
  \simat{p} v'$.
\end{lemma}
\begin{proof}
  For the first part, we proceed by induction on the total size of $p,p'$.  The cases where one of $p,p'$ is $\vhole$ or $\vany$ are straightforward; for $\vany$ we also need Lemma~\ref{lem:sub-sim}.  The cases involving constants, pairs, or complete set and record patterns are also straightforward.

There are several similar cases involving partial set or record patterns.  We illustrate two representative cases:
\begin{enumerate}
\item If $p = \vset{\overline{\lbl_i.p_i},\overline{\lbl'_i.q_i}}$ and $p' =
  \vset{\overline{\lbl_i.p_i'}}\udot \vhole$ then $p \sqcup p' =
  \vset{\overline{\lbl_i.p_i\sqcup p_i'},\overline{\lbl'_i.q_i}}$.  First,
  suppose $v \simat{p \sqcup p'} v'$.  This means $v =
  \vset{\overline{\lbl_i.v_i},\overline{\lbl'_i.w_i}}$ and $v' =
  \vset{\overline{\lbl_i.v_i'},\overline{\lbl'_i.w_i'}}$, where $v_i
  \simat{p_i \sqcup p_i'} v_i'$ and $w_i \simat{q_i} w_i'$.
  Therefore, by induction, $v_i \simat{p_i} v_i'$ and $v_i
  \simat{p_i'} v_i'$, so we can conclude that
  $\vset{\overline{\lbl_i.v_i},\overline{\lbl'_i.w_i}} \simat{p}
  \vset{\overline{\lbl_i.v_i'},\overline{\lbl'_i.w_i'}}$ and
  $\vset{\overline{\lbl_i.v_i},\overline{\lbl'_i.w_i}}
  \simat{p'}\vset{\overline{\lbl_i.v_i'},\overline{\lbl'_i.w_i'}}$, as
  required.

  Conversely, if we assume $v \simat{p} v'$ and $v \simat{p'} v'$,
  then we must have $v =
  \vset{\overline{\lbl_i.v_i},\overline{\lbl'_i.w_i}}$ and $v' =
  \vset{\overline{\lbl_i.v_i'},\overline{\lbl'_i.w_i'}}$, where $v_i
  \simat{p_i} v_i'$ and $v_i \simat{p_i'} v_i'$ and $w_i \simat{q_i}
  w_i'$.  Thus, by induction we have $v_i \simat{p_i \sqcup p_i'}
  v_i'$ so we can conclude
  $\vset{\overline{\lbl_i.v_i},\overline{\lbl'_i.w_i}}\simat{p\sqcup p'}
  \vset{\overline{\lbl_i.v_i'},\overline{\lbl'_i.w_i'}}$.
\item If $p = \vset{\overline{\lbl_i.p_i},\overline{\lbl'_i.q_i}}\udot
  \vhole$ and $p' =
  \vset{\overline{\lbl_i.p_i'},\overline{\lbl''_i.r_i}}\udot \vany$ then
  $p \sqcup p' = \vset{\overline{\lbl_i.p_i\sqcup
      p_i'},\overline{\lbl'_i.q_i[\vany/\vhole]},\overline{\lbl''_i.r_i}}
  \udot \vany$.  First, suppose $v \simat{p \sqcup p'} v'$.  This
  means $v =
  \vset{\overline{\lbl_i.v_i},\overline{\lbl'_i.w_i},\overline{\lbl_i''.u_i}} \semu w_0$
  and $v' =
  \vset{\overline{\lbl_i.v_i'},\overline{\lbl'_i.w_i'},\overline{\lbl_i''.u_i'}} \semu w_0$,
  where $v_i \simat{p_i \sqcup p_i'} v_i'$ and $w_i
  \simat{q_i[\vany/\vhole]} w_i'$ and $u_i \simat{r_i} u_i'$.
  Therefore, by induction, $v_i \simat{p_i} v_i'$ and $v_i
  \simat{p_i'} v_i'$, and we also have $w_i \simat{q_i} w_i'$ and $w_i
  = w_i'$, so we can conclude that
  $\vset{\overline{\lbl_i.v_i},\overline{\lbl'_i.w_i},\overline{\lbl_i''.u_i}} \semu w_0
  \simat{p} \vset{\overline{\lbl_i.v_i'},\overline{\lbl'_i.w_i'},\overline{\lbl_i''.u_i'}}\semu w_0$ and
  $\vset{\overline{\lbl_i.v_i'},\overline{\lbl'_i.w_i},\overline{\lbl''_i.u_i}}\semu w_0
  \simat{p'}\vset{\overline{\lbl_i.v_i'},\overline{\lbl'_i.w_i'},\overline{\lbl_i''.u_i'}}\semu w_0$,
  as required.

  Conversely, if we assume $v \simat{p} v'$ and $v \simat{p'} v'$,
  then we must have $v =
  \vset{\overline{\lbl_i.v_i},\overline{\lbl'_i.w_i},\overline{\lbl''_i.u_i}}\semu
  v_0$ and $v' =
  \vset{\overline{\lbl_i.v_i'},\overline{\lbl'_i.w_i'},\overline{\lbl''_i.u_i'}}
  \semu v_0'$ where $v_i \simat{p_i} v_i'$ and $v_i \simat{p_i'} v_i'$
  and $w_i \simat{q_i} w_i'$ and $u_i \simat{r_i} u_i'$.  In addition,
  we must have that $w_i = w_i'$ and $v_0 = v_0'$ since $v$ and $v'$
  must be equal at all labels not in
  $\overline{\lbl_i},\overline{\lbl_i''}$.  Thus, by induction we have
  $v_i \simat{p_i \sqcup p_i'} v_i'$ and (using Lemma~\ref{lem:sub-sim}) we can also easily show that
  $w_i \simat{q_i[\vany/\vhole]} w_i$, so we can conclude that $
  \vset{\overline{\lbl_i.v_i},\overline{\lbl'_i.w_i},\overline{\lbl_i''.u_i'}}\semu
  v_0\simat{p \sqcup p'}
  \vset{\overline{\lbl_i.v_i'},\overline{\lbl_i'.w_i},\overline{\lbl''_i.u_i}}
  \semu v_0$.
\end{enumerate}

The second part follows immediately from the definition of $p \sqleq p'$ as $p \sqcup p' = p'$.
\end{proof}

\begin{lemma}[Properties of union and restriction]\label{lem:union-restriction}
~
\begin{enumerate}
\item If $p_1 \sqleq v_1$ and $p_2 \sqleq v_2$ and $v_1 \simat{p_1}
    v_1'$ and $v_2 \simat{p_2} v_2'$ then $v_1 \semu v_2 \simat{p_1
      \semu p_2} v_1' \semu v_2'$, provided all of these disjoint unions
    are defined.
\item If $p \sqleq v_1 \semu v_2$ and $L_1 \leq \dom(v_1)$ and
    $L_2 \leq \dom(v_2)$ and $L_1,L_2$ are prefix-disjoint, then $p|_{L_1} \sqleq v_1$ and
    $p|_{L_2} \sqleq v_2$.
\end{enumerate}
\end{lemma}
\begin{proof}
  For part 1, 
  assume $p_1 \sqleq v_1$, $p_2 \sqleq v_2$, $v_1 \simat{p_1}
    v_1'$ and $v_2 \simat{p_2} v_2'$, and assume that the domains of
    $p_1$ and $p_2$, $v_1$ and $v_2$, and $v_1'$ and $v_2'$ are
    prefix-disjoint respectively, so that the unions exist.  There are
    several cases.  If $p_1$ or $p_2$ is $\vhole$ then the conclusion
    is immediate.  If both are $\vany$ then $v_1 = v_1'$ and $v_2 =
    v_2'$ so $v_1 \semu v_2 = v_1' \semu v_2'$.  

Most of the remaining
    cases are straightforward; we illustrate with the case $p_1 =
    \vset{\overline{\lbl_i.p_i}} \udot \vhole$ and $p_2 =
    \vset{\overline{\lbl_i'.q_i}} \udot \vany$.  In this case, 
    \begin{eqnarray*}
     v_1 &=&    \vset{\overline{\lbl_i.v^1_i}} \semu v^1_0\\
     v_2 &=&    \vset{\overline{\lbl_i'.v^2_i}} \semu v^2_0\\
     v_1' &=&    \vset{\overline{\lbl_i.w^1_i}} \semu w^1_0\\
     v_2' &=&    \vset{\overline{\lbl_i'.w^2_i}} \semu w^2_0
    \end{eqnarray*}
    and we
    also know that $v^1_i\simat{p_i} w^1_i$ and
    $v^2_i \simat{q_i} w^2_i$ for each $i$.  Therefore, 
 
   \begin{eqnarray*}
\vset{\overline{\lbl_i.v^1_i}} \semu v^1_0 \semu \vset{\overline{\lbl_i'.v^2_i}} \semu v^2_0 
 &=&  \vset{\overline{\lbl_i.v^1_i},\overline{\lbl'_i.v^2_i}} \semu v^1_0  \semu v^2_0
\\
&\simat{\vset{\overline{\lbl_i.p_i},\overline{\lbl_i'.q_i}} \udot \vhole} &  \vset{\overline{\lbl_i.w^1_i},\overline{\lbl'_i.w^2_i}} \semu w^1_0  \semu w^2_0\\
&=&
\vset{\overline{\lbl_i.w^1_i}} \semu w^1_0 \semu \vset{\overline{\lbl_i'.w^1_i}} \semu w^2_0
    \end{eqnarray*}

    For part 2, assume $p \sqleq v_1 \semu v_2$ and $L_i \leq
    \dom(v_i)$ for $i \in \{1,2\}$.  We proceed by case analysis on
    $p$.  The cases $p = \vhole$ and $p = \vany$ are immediate since
    $\vhole|_L = \vhole$ and $\vany|_L = \vany$.  If $p =
    \vset{\overline{\lbl_i.p_i}}$ is a complete set pattern, then
    $p|_{L_1}$ selects just those elements of $p$ that have a prefix
    in $L_i$, and since every element of $v_q $ has its label's prefix
    in $L_1$ we must have that $v_1 = \vset{\overline{\lbl'_i.v'_i}} $ where
    $p_i \sqleq v_i'$ for each $i$ and $ p|_{L_1} =
    \vset{\overline{\lbl_i'.p_i'}}$, which is what we need to show.
    The cases for $p = \{\overline{\lbl_i.p_i}\} \udot \vhole$ and $p
    = \{\overline{\lbl_i.p_i}\} \udot \vany$ are similar.

A symmetric argument suffices to show $p|_{L_2} \sqleq v_2$.
\end{proof}

\begin{lemma}[Projection and $\sqleq$]\label{lem:projection-sub}
~
  \begin{enumerate}
  \item If $p \sqleq \set{\epsilon.v}$ then $p.\epsilon \sqleq v$.
  \item If $p \sqleq 1\cdot v_1 \semu 2 \cdot v_2$ then $p[1] \sqleq
    v_1$ and $p[2] \sqleq v_2$.
  \item If $p \sqleq \lbl\cdot v $ then $p[\lbl] \sqleq v$.
  \item If $p \sqleq \vrec{\overline{A_i:v_i}}$ then $p.A_i \sqleq
    v_i$.
 \end{enumerate}
\end{lemma}

\begin{proof}
  For part (1), suppose $p \sqleq \vset{\epsilon.v}$.  We proceed by case analysis on $p$.  The cases for $\vhole$ and $\vany$ are trivial.  If $p = \{\epsilon.p'\}$ then $p.\epsilon = p'$ so the conclusion follows.  The cases for $p = \{\epsilon.p'\} \udot \vhole$ or $p = \{\epsilon.p'\} \udot \vany$ are similar.

For part (2), suppose $p \sqleq 1\cdot v_1 \semu 2\cdot v_2$.  Suppose  $v_1 = \{\overline{\lbl_i.v_i}\}$ and $v_2 = \{\overline{\lbl_i'.v_i'}\}$.  If $p$ is a complete set pattern, then it must be of the form $ \{\overline{1.\lbl_i.p_i},  \overline{2.\lbl_i'.q_i}\}$  , where $p_i \sqleq v_i$ and $q_i \sqleq v_i'$.  The desired conclusion follows since $p[1 ] = \{ \overline{\lbl_i.p_i}\}$, and a symmetric argument shows that $p[2] \sqleq v_2$.  The cases for partial set patterns are similar, since the $\vhole$ or $\vany$ is preserved by the projection operation.

For part (3), suppose $p \sqleq \lbl\cdot v$.  The proof is analogous
to the previous case.

For part (4), suppose $p \sqleq \vrec{\overline{A_i:v_i}}$.  If $p = \vany $ or $\vhole$, the result is immediate; otherwise, $p$ is a record pattern.  If it is a total record pattern $\vrec{\overline{A_i:p_i}}$ then clearly $p.A_i = p_i \sqleq v_i$.  Otherwise, it is a partial pattern, in which case either $p.A_i$ is a pattern $p_i$ mentioned in $p$, in which case we are done, or $p.A_i = \vany$ or $p.A_i = \vhole$, and the conclusion follows immediately.
\end{proof}

\begin{lemma}[Projection and $\simat{p}$]\label{lem:projection-sim}
~
 \begin{enumerate}
  \item If $p \sqleq \{\epsilon.v\}$ and $v \simat{p.\epsilon} v'$ then $\set{\epsilon.v} \simat{p}
    \set{\epsilon.v'}$.  
  \item If $p \sqleq 1 \cdot v_1 \semu 2 \cdot v_2$ and $v_1 \simat{p[1]}
    v_1'$ and $v_2 \simat{p[2]} v_2'$ then $1 \cdot v_1 \semu 2 \cdot
    v_2 \simat{p} 1 \cdot v_1' \semu 2 \cdot v_2'$.
  \item If $p \sqleq \lbl \cdot v$ and $v \simat{p[\lbl]} v'$ then
    $\lbl\cdot v \simat{p} \lbl \cdot v'$.
\item If $p \sqleq \vrec{\overline{A_i:v_i}}$ and $v_1\simat{p.A_1}
    v_1', \ldots, v_n \simat{p.A_n} v_n'$ then $\vrec{\overline{A_i:v_i}} \simat{p} \vrec{\overline{A_i:v_i'}}$.
  \end{enumerate}
\end{lemma}
\begin{proof}
  For part (1), suppose $p \sqleq \vset{\epsilon.v}$ and $v \simat{p.\epsilon} v'$.  If $p = \vhole$ or $\vany$ then the result is immediate (this is the case for all three parts of the lemma).  If $p = \{\epsilon.p'\}$, $p = \{\epsilon.p'\}\udot \vhole$, or $p = \{\epsilon.p'\} \udot \vany$ then $p.\epsilon = p'$ so $v \simat{p'} v'$.  We can conclude $\vset{\epsilon.v} \simat{p} \vset{\epsilon.v'}$.

For part (2), suppose $p \sqleq 1\cdot v_1 \semu 2\cdot v_2$ and $v_1 \simat{p[1]} v_1'$ and $v_2 \simat{p[2]} v_2'$.  As usual, the cases $p = \vhole$ and $p= \vany$ are trivial.  Suppose  $v_1 = \{\overline{\lbl_i.v'_i}\}$ and $v_2 = \{\overline{\lbl_i'.v_i''}\}$.  If $p$ is a complete set pattern it must be of the form $\{\overline{1.\lbl_i.p_i},  \overline{2.\lbl_i'.q_i}\}$, and $v_1 \simat{\vset{\overline{\lbl_i.p_i}}} v_1'$ and $v_2 \simat{\vset{\overline{\lbl_i'.q_i}}} v_2'$.  It is straightforward to show that $1 \cdot v_1 \simat{\{\overline{1 \cdot\lbl_i.p_i}\}} 1 \cdot v_1'$ and $2 \cdot v_2 \simat{\{\overline{2 \cdot\lbl_i'.q_i}\}} 2 \cdot v_2'$, so by previous results we have $1 \cdot v_1 \semu 2 \cdot v_2 \simat{p} 1 \cdot v_1' \semu 2 \cdot v_2'$.

The cases for partial patterns follow the same reasoning, making use of the fact that projection preserves the partial pattern.

For part (3), suppose $p \sqleq \lbl\cdot v$.  The argument is similar
to the previous case.

For part (4) suppose $p \sqleq \vrec{\overline{A_i:v_i}}$ and $v_i \eqat{p.A_i} v_i'$ for each $i$.  If $p$ is $\vhole$, $\vany$, or a complete record pattern then the conclusion is immediate.  Otherwise, if $p = \vrec{\overline{B_i:q_i}} \udot \vhole$, the conclusion is immediate since each component of the records $\vrec{\overline{A_i:v_i}}$ and $\vrec{\overline{A_i:v_i'}}$ either match the appropriate $q_i$ or need not match because of the hole.  Finally, if $p = \vrec{\overline{B_i:q_i}} \udot \vany$, the conclusion follows since each pair of corresponding components of the records $\vrec{\overline{A_i:v_i}}$ and $\vrec{\overline{A_i:v_i'}}$ either match the appropriate $q_i$ or are equal because $p.A_i = \vany$ if $A_i$ is not among the $B_j$.
\end{proof}

\section{Proof of correctness of trace slicing}
\label{app:proof-trace-slicing}

\begin{lemma}
  \label{lem:subset-replay}
If $\gamma,x \in v_1,\Theta \trrun^* v_2$ and $v_1' \subseteq v_1$
then there exists $v_2' \subseteq v_2$ such that $\gamma,x \in v_1', \Theta \trrun^* v_2'$.
\end{lemma}
\begin{proof}
  The proof is straightforward by induction on derivations.  The cases
  for $v_1 = \emptyset$ and $v_1 = \{\lbl.v\}$ are immediate; if $v_1
  = w_1 \semu w_2$ then we proceed by induction using the subsets
  $w_1' = w_1 \cap v_1'$ and $w_2' = w_2 \cap v_1'$.
\end{proof}

We prove correctness of the full trace slicing algorithm, with
enriched patterns, since the
correctness for simple patterns follows as a special case.
  \begin{theorem}[Correctness of Slicing]~
    \label{thm:trace-slicing-full}
    \begin{enumerate}
    \item Suppose $\gamma,T \trrun v$ and $p \sqleq v$ and $p, T
      \trbslice \rho, S$.  Then for all $\gamma' \simat{\rho} \gamma$
      and $T' \sqgeq S$ such that $\gamma',T' \trrun v'$ we have $v'
      \simat{p} v$.

    \item Suppose $\gamma,x\in v_0, \Theta \trrun^* v$ and $p \sqleq
      v$ and $p, x.\Theta_0 \trbslice^* \rho, \Theta_0', p_0$, where
      $\Theta_0 \subseteq \Theta$.  Then for all $\gamma'
      \simat{\rho}\gamma$ and $v_0' \simat{p_0} v_0$ and
      $\Theta'\sqgeq \Theta_0'$ such that $\gamma',x\in v_0',\Theta'
      \trrun^* v'$ we have $v' \simat{p} v$.
    \end{enumerate}
  \end{theorem}
\begin{proof}
  For part (1), the proof is by induction on the structure of slicing derivations, using inversion to extract information from other derivations.  The cases for variables, constants, primitive operations, and let-binding are exactly as in previous work~\cite{acar13jcs}.  The cases for conditionals are similar to the those for variant types and case constructs in previous work.

We show the cases for records and set operations, which are new to
this paper (records are handled similarly to pairs in our
previous work, so the cases for pairs are omitted).
  \begin{itemize}
\item Field projection.  If the last step in the slicing derivation is 
\[\inferrule*{
\vrec{A_i{:}p;\vhole},T \trbslice \rho,S
}{
p,\trfield{T}{A_i} \trbslice \rho,\trfield{S}{A_i}
}
\]
then the evaluation derivation must be of the form
\[
\inferrule*{
\gamma,T \trrun \kwrec{A_1{:}v_1,\ldots,A_n{:}v_n}
}{
\gamma,\kwfield{T}{A_i} \trrun v_i
}
\]
Let $\gamma' \simat{\rho} \gamma$ and $T' \sqgeq\trfield{S}{A_i}$ be
given, where $\gamma',T' \trrun v'$.  Then $T'$ must have the form
$\trfield{T''}{A_i}$ for some $T'' \sqgeq S$ so the replay derivation is of the form
\[
\inferrule*{
\gamma',T'' \trrun \kwrec{A_1{:}v_1',\ldots,A_n{:}v_n'}
}{
\gamma',\kwfield{T''}{A_i} \trrun v_i'
}
\]
The induction hypothesis applies since it is easy to show that
$\vrec{A_i:p;\vhole} \sqleq  \kwrec{A_1{:}v_1,\ldots,A_n{:}v_n}$.
Therefore 
\[
\kwrec{A_1{:}v_1',\ldots,A_n{:}v_n'}
\simat{\vrec{A_i:p;\vhole}}
\kwrec{A_1{:}v_1,\ldots,A_n{:}v_n}
\;.\]
From this it is obvious that $v_i' \simat{p} v_i$.
\item Record. If the last step in the slicing derivation is 
\[\inferrule*
{
p.A_1,T_1 \trbslice \rho_1,S_1\\
\cdots\\
p.A_n,T_n \trbslice \rho_n,S_n
}
{
p, \trrec{A_1{:}T_1,\ldots,A_n{:}T_n}
\trbslice \rho_1 \sqcup \cdots \sqcup \rho_n,
\kwrec{A_1{:}S_1,\ldots,A_n{:}S_n}
}\]
then the evaluation derivation must be of the form
\[
\inferrule*{
\gamma,T_1 \trrun v_1\\
\cdots\\
\gamma,T_n \trrun v_n
}{
\gamma,\kwrec{A_1{:}T_1,\ldots,A_n{:}T_n} \trrun \kwrec{A_1{:}v_1,\ldots,A_n{:}v_n}
}
\]
Let $\gamma' \simat{\rho_1 \sqcup \cdots \sqcup \rho_n} \gamma$ and
$T' \sqgeq\kwrec{A_1{:}S_1,\ldots,A_n{:}S_n}$ be given, where
$\gamma',T' \trrun v'$.  Then $T'$ must have the form $
\trrec{A_1{:}T_1',\ldots,A_n{:}T_n'}$, where $T_i' \sqgeq S_i$ for
each $i$, so the replay derivation is of the form
\[
\inferrule*{
\gamma',T_1' \trrun v_1'\\
\cdots\\
\gamma',T_n' \trrun v_n'
}{
\gamma',\kwrec{A_1{:}T_1',\ldots,A_n{:}T'_n} \trrun \kwrec{A_1{:}v_1',\ldots,A_n{:}v_n'}
}
\]
By Lemma~\ref{lem:projection-sub} we know $p.A_i \sqleq v_i$ for each $i$, and $\gamma'\simat{\rho_i} \gamma$ for each $i$, so by induction $v_i' \simat{p.A_i} v_i$ for each $i$.  Using Lemma~\ref{lem:projection-sim} we can conclude that $\kwrec{A_1{:}v_1',\ldots,A_n{:}v_n'} \simat{p} \kwrec{A_1{:}v_1,\ldots,A_n{:}v_n}$.
  \item Empty set.  This case is trivial, similar to the usual case for constants.
  \item Singleton.  This case follows immediately from the relevant properties of $p.\epsilon$, using similar reasoning to the record projection case.
  \item Union.  If the last step in the slicing derivation is
\[
\inferrule*
{
p[1], T_1 \trbslice \rho_1, S_1 \\
p[2], T_2 \trbslice \rho_2, S_2 
}
{
 p, \trsetu{T_1}{T_2} \trbslice \rho_1 \sqcup \rho_2 , \trsetu{S_1}{S_2}
}\]
then the evaluation derivation must be of the form
\[
\inferrule*
{
\gamma, T_1 \trrun v_1
\\
\gamma, T_2 \trrun v_2
\\
}
{
\gamma,\trsetu{T_1}{T_2} \trrun 1 \cdot v_1 \semu 2 \cdot v_2
}
\]
Let $\gamma' \simat{\rho_1 \sqcup \rho_2} \gamma$ and $T'
\sqgeq\trsetu{T_1}{T_2}$ be given, where $\gamma',T' \trrun v'$.  Then
$T'$ must have the form $\trsetu{T_1'}{T_2'} $ where $T_i' \sqgeq S_i$ so the replay
derivation is of the form
\[
\inferrule*
{
\gamma', T_1' \trrun v_1'
\\
\gamma', T_2' \trrun v_2'
\\
}
{
\gamma',\trsetu{T_1'}{T_2'} \trrun 1 \cdot v_1' \semu 2 \cdot v_2'
}
\]
By Lemma~\ref{lem:projection-sub} we know $p[1] \sqleq v_1$ and
$p[2]\sqleq v_2$, and $\gamma' \simat{\rho_i} \gamma$, so by induction
we have $v_1' \simat{p[1]} v_1$ and $v_2' \simat{p[2]} v_2$, and
therefore by Lemma~\ref{lem:projection-sim} we can conclude $1 \cdot
v_1' \semu 2\cdot v_2'\simat{p}1 \cdot v_1 \semu 2 \cdot v_2$.

\item Sum and emptiness.  In both cases, since the slice ensures that
  the whole argument to the sum or emptiness test is
  preserved, the argument is straightforward.  For example, for
  emptiness suppose the derivation is of the form:
\[
\inferrule*{
\vany, T \trbslice \rho,S
}{
p,\trsetise{T} \trbslice \rho,\trsetise{S}
}
\]
Then there are two cases. If replay derivation is of the form
\[
\inferrule*{
\gamma,T \trrun \emptyset
}{
\gamma,\trsetise{T} \trrun \kwtrue
}
\]
Suppose $\gamma' \simat{\rho} \gamma$ and $T' \sqgeq \trsetise{S}$
with $\gamma',T' \trrun v'$. Then $T'$ is of the form $\trsetise{T''}$
with $T'' \sqgeq S$, so the replay derivation must be of the form:
\[
\inferrule*{
\gamma',T'' \trrun v''
}{
\gamma',\trsetise{T''} \trrun v'
}
\]
By induction, $v'' \simat{\vany} \emptyset$, which implies $v'' =
\emptyset$ so $v' = \kwtrue \simat{p} \kwtrue$ (since $p \sqleq
\kwtrue$).  The cases where the argument to $\kwsetise{T}$ evaluates
to a nonempty set, and for $\kwsum{T}$, are similar.

\item Comprehension.  If the derivation is of the form
\[
\inferrule*{
 p, x.\Theta \trbslice^* \rho', \Theta_0, p_0\\
p_0, T \trbslice \rho,S
}{
p,\trcomp{e}{x}{T}{\Theta} \trbslice \rho \sqcup \rho', \trcomp{e}{x}{S}{\Theta_0}
}
\]
then the replay derivation must be of the form 
\[
\inferrule*
{
\gamma, T \trrun v_0 \\
 \gamma,x\in v_0 \Theta \trrun^* v
}
{
\gamma, \trcomp{e}{x}{T}{\Theta}
\trrun
v
}
\]
Suppose $\gamma' \simat{\rho\sqcup \rho'} \gamma$ and $T' \sqgeq
\trcomp{e}{x}{S}{\Theta_0}$ with $\gamma',T' \trrun v'$.  Then $T'$
must be of the form $\trcomp{e}{x}{T''}{\Theta'}$ for some $T'' \sqgeq
S$ and $\Theta' \sqgeq \Theta_0$ so the replay derivation must be of the form
\[
\inferrule*
{
\gamma', T'' \trrun v_0' \\
 \gamma',x\in v_0', \Theta' \trrun^* v'
}
{
\gamma', \trcomp{e}{x}{T''}{\Theta'}
\trrun
v'
}
\]
Since $p_0 \sqleq v_0$ and $\gamma' \simat{\rho}\gamma$ and $\gamma'
\simat{\rho'} \gamma$ we have $v_0' \simat{p_0} v_0$ by induction.
Thus, by the second induction hypothesis, since $p \sqleq v$ and
$\gamma' \simat{\rho'} \gamma$ and $v_0' \simat{p_0} v_0$ we have $v' \simat{p}
v$. 
  \end{itemize}

For part (2), the proof is again by induction on the structure of derivations.
\begin{itemize}
\item If the slicing derivation is of the form
\[
\inferrule*
{
\strut
}
{
\emptyset, x.\emptyset \trbslice^* [], \emptyset, \emptyset
}
\]
 then the conclusion is immediate, since rerunning an empty trace set always yields the empty set.
\item If the slicing derivation is of the form
\[
\inferrule*
{
\strut
}
{
\vany, x.\emptyset \trbslice^* [], \emptyset, \emptyset
}
\]
 then the conclusion is immediate as before, since rerunning an empty trace set always yields the empty set. 
\item If the slicing derivation is of the form
\[
\inferrule*
{
\strut
}
{
\vhole, x.\Theta \trbslice^* [], \emptyset, \vhole
}
\]
 then the conclusion is immediate, since any two values match according to $\vhole$.

\item If the slicing derivation is of the form
\[
\inferrule*
{
p[\lbl], T \trbslice \rho[x\mapsto p_0],S
}
{
p, x.\{\lbl.T\} \trbslice^* \rho, \{\lbl.S\}, \{\lbl.p_0\}
}
\]
then observe that $\Theta \supseteq \{\lbl.T\}$ by assumption, so
$\Theta(\lbl) = T$.
 So, the replay derivation must have the form 
\[\inferrule*
{
 \lbl \in \dom(\Theta) \\
 \gamma[x\mapsto v_0],T \trrun v 
}
{
 \gamma,x\in \set{\lbl.v_0}, \Theta \trrun^* \lbl \cdot v
}
\]
Now suppose that $\gamma' \simat{\rho} \gamma$ and $v_0'
\simat{\set{\lbl.p_0}} \set{\lbl.v_0}$ and $\Theta' \sqgeq \{\lbl.S\}$
are given where $\gamma',x\in v_0',\Theta' \trrun^* v'$.  It follows
that $v_0' = \{\lbl.v_0''\}$ and $v_0'' \simat{p_0} v_0$, so
$\gamma'[x\mapsto v_0''] \simat{\rho[x\mapsto p_0]} \gamma[x\mapsto
v_0]$.  Moreover, by inversion the derivation must have the form:
\[\inferrule*
{
 \lbl \in \dom(\Theta') \\
 \gamma'[x\mapsto v_0''],\Theta'(\lbl) \trrun v'
}
{
 \gamma',x\in \set{\lbl.v_0''}, \Theta' \trrun^* \lbl \cdot v'
}
\]
Since by Lemma~\ref{lem:projection-sub} $p[\lbl] \sqleq v$ and $\Theta'(\lbl) \sqgeq S$ (which holds because $\Theta' \sqgeq
\{\lbl.S\}$), we have by induction that $v' \simat{p[\lbl]} v$, and using Lemma~\ref{lem:projection-sim} we can conclude that $\lbl \cdot v'  \simat{p} \lbl \cdot v$, as desired.
\item Suppose the slicing derivation is of the form:
\[
\inferrule*
{
p|_{\dom(\Theta_1)}, x.\Theta_1 \trbslice^* \rho_1,\Theta_1',p_1\\
p|_{\dom(\Theta_2)}, x.\Theta_2 \trbslice^* \rho_2,\Theta_2',p_2
}
{
p, x.\Theta_1 \semu \Theta_2 \trbslice^* \rho_1 \sqcup \rho_2,
\Theta_1' \semu \Theta_2', p_1 \semu p_2
}
\]
and suppose that $\gamma, x \in v_0, \Theta \trrun^* v$ where $\Theta
\supseteq \Theta_1 \semu \Theta_2$.
Suppose that $\gamma' \simat{\rho_1 \sqcup \rho_2} \gamma$ and $v_0'
\simat{p_1 \semu p_2} v_0$ and $\Theta' \sqgeq \Theta_1' \semu
\Theta_2'$ are given, where $\gamma', x \in v_0', \Theta' \trrun^*
v'$.  We need to show that $v' \simat{p} v$.  

Since $v_0' \simat{p_1 \semu p_2} v_0$, it is straightforward to show
that there must exist $v_1, v_2, v_1', v_2'$ such that $v_1 \semu v_2 =
v_0$, $v_1' \semu v_2' = v_0'$, $v_1 \simat{p_1} v_1$ and $v_2'
\simat{p_2} v_2$.  Furthermore, by Lemma~\ref{lem:subset-replay} we
know that $\gamma,x \in v_i, \Theta \trrun^* w_i$ and $\gamma',x \in
v_i', \Theta' \trrun^* w_i'$ for some $w_1,w_2,w_1',w_2'$.  Therefore,
we can conclude that:

\[
\inferrule*
{
\gamma,x\in v_1,\Theta \trrun^* w_1\\
\gamma,x\in v_2,\Theta \trrun^* w_2\\
}
{
\gamma,x\in v_1\semu v_2,\Theta \trrun^* w_1\semu w_2\\
}
\qquad
\inferrule*
{
\gamma',x\in v_1',\Theta' \trrun^* w_1'\\
\gamma',x\in v_2',\Theta' \trrun^* w_2'\\
}
{
\gamma',x\in v_1'\semu v_2',\Theta' \trrun^* w_1'\semu w_2'\\
}
\]
Furthermore, since $v_1 \semu v_2 = v_0$, by determinacy we know that
$w_1 \semu w_2 = v$ and similarly $w_1' \semu w_2 = v'$.  By induction
since $p_i \sqleq v_i$ and $\Theta_i \subseteq \Theta$,
we know that $w_i' \simat{p|_{\dom(\Theta_i)}} w_i$, so we know that
\[v' = w_1' \semu w_2' \simat{p|_{\dom(\Theta_i)} \semu
  p|_{\dom(\Theta_2)}} w_1 \semu w_2 = v\;.\]

To conclude, since $p \sqleq w_1 \semu w_2$ and $\dom(\Theta_1) \leq
w_1$ and $\dom(\Theta_2) \leq w_2$, it follows that $p =
p|_{\dom(\Theta_1)}\semu p|_{\dom(\Theta_2)}$, so we can conclude $w_1
\semu w_2 \simat{p} w_1' \semu w_2'$ as desired.
\end{itemize}
This exhausts all cases and completes the proof.
\end{proof}

\section{Proof of correctness of query slicing}
\label{app:proof-query-slicing}

\begin{lemma}
  \label{lem:subset-eval}
If $\gamma,x \in v_1,e\red^* v_2$ and $v_1' \subseteq v_1$
then there exists $v_2' \subseteq v_2$ such that $\gamma,x \in v_1',e \red^* v_2'$.
\end{lemma}
\begin{proof}
  The proof is straightforward by induction on derivations.  The cases
  for $v_1 = \emptyset$ and $v_1 = \{\lbl.v\}$ are immediate; if $v_1
  = w_1 \semu w_2$ then we proceed by induction using the subsets
  $w_1' = w_1 \cap v_1'$ and $w_2' = w_2 \cap v_1'$.
\end{proof}

We prove Theorem~\ref{thm:query-slicing} by strengthening the
induction hypothesis as follows:
\begin{theorem}[Correctness of Query Slicing]\label{thm:query-slicing-full}
  ~
\begin{enumerate}
\item 
    Suppose $\gamma,T \trrun v$ and $p \sqleq v$ and $p, T \uneval \rho, e$.  Then
    for all $\gamma' \simat{\rho} \gamma$ and $e'\sqgeq e$ such that
    $\gamma',e' \red v'$ we have $v' \simat{p} v$.

  \item Suppose $\gamma,x\in v_0, \Theta \trrun^* v$ and $p \sqleq v$ and $p,
    \Theta_0 \uneval^* \rho, e_0, p_0$, where
      $\Theta_0 \subseteq \Theta$.  Then for all $\gamma'
    \simat{\rho}\gamma$ and $v_0' \simat{p_0} v_0$ and $e_0'\sqgeq
    e_0$ such that $\gamma',x\in v_0', e_0' \red^* v'$ we have $v'
    \simat{p} v$.
  \end{enumerate}
\end{theorem}
\begin{proof}
  The proof is by induction on the structure of query slicing
  derivations.  Many of the cases are essentially the same as for
  trace slicing.  The cases for conditionals are straightforward,
  since in either case the sliced trace and environment retain enough
  information to force the same branch to be taken on recomputation.
  We show the details of the cases involving conditionals and comprehensions.

For part (1), we consider a conditional and  comprehension rule:
\begin{itemize}
\item 
If the slicing derivation is of the form:
\[
\inferrule*
{
p_1,T_1 \uneval \rho_1,e_1' \\
\kwtrue,T \uneval \rho,e'
}
{
p_1,\trifthen{T}{e_1}{e_2}{T_1} \uneval \rho_1\sqcup \rho, \kwif{e'}{e_1'}{\vhole}
}
\]
then the replay derivation must be of the form:
\[
\inferrule*{
\gamma,T \trrun \kwtrue \\
\gamma,T_1 \trrun v_1
}{
\gamma,\trifthen{T}{e_1}{e_2}{T_1} \trrun v_1
}\]
Suppose $\gamma'\simat{\rho_1\sqcup \rho}\gamma$ and $e'' \sqgeq
\kwif{e'}{e_1'}{\vhole}$ are given where $\gamma',e' \red v'$.  Then
$e'' = \kwif{e_0''}{e_1''}{e_2''}$, where $e_0'' \sqgeq e'$ and $e_1''
\sqgeq e_1'$, so there are two cases for the evaluation derivation.  If it has the form
\[
\inferrule*{
\gamma',e_0' \red \kwfalse \\
\gamma',e_2'\red v_2'
}{
\gamma',\kwif{e_0'}{e_1'}{e_2'} \red v_2'
}
\]
then since $\gamma'\simat{\rho} \gamma$ and $\gamma' \simat{\rho_1} \gamma$, by induction we would have that $\kwfalse \simat{\kwtrue}\kwtrue$, which is absurd.  So this case cannot arise.

Otherwise, the derivation must have the form:
\[
\inferrule*{
\gamma',e_0' \red \kwtrue \\
\gamma',e_1' \red v_1'
}{
\gamma',\kwif{e_0'}{e_1'}{e_2'} \red v_1'
}
\]
Since $\gamma'\simat{\rho} \gamma$ and $\gamma' \simat{\rho_1} \gamma$, by induction we have that $v_1' \simat{p_1} v_1$ as desired.

\item Comprehension.  If the derivation is of the form
\[
\inferrule*{
 p, x.\Theta \uneval^* \rho', e_1', p_0\\
p_0, T \uneval \rho,e_0'
}{
p,\trcomp{e}{x}{T}{\Theta} \uneval \rho \sqcup \rho', \kwcomp{e_1'}{x}{e_0'}
}
\]
then the replay derivation must be of the form 
\[
\inferrule*
{
\gamma, T \trrun v_0 \\
 \gamma,x\in v_0, \Theta \trrun^* v
}
{
\gamma, \trcomp{e}{x}{T}{\Theta}
\trrun
v
}
\]
Suppose $\gamma' \simat{\rho\sqcup \rho'} \gamma$ and $e'' \sqgeq
\kwcomp{e_1'}{x}{e_0'}$ with $\gamma',e'' \red v'$.  Then $e''$ must be
of the form $\kwcomp{e_1''}{x}{e_0''}$ for some $e_1'' \sqgeq e_1'$ and
$e_0'' \sqgeq e_0'$, so the evaluation derivation must be of the form
\[
\inferrule*
{
\gamma', e_0'' \red v_0' \\
 \gamma',x\in v_0', e_1'' \red^* v'
}
{
\gamma', \kwcomp{e_1''}{x}{e_0''}
\red
v'
}
\]
Since $p_0 \sqleq v_0$ and $\gamma' \simat{\rho}\gamma$ and $\gamma'
\simat{\rho'} \gamma$ we have $v_0' \simat{p_0} v_0$ by induction.
Thus, by the second induction hypothesis, since $p \sqleq v$ and
$\gamma' \simat{\rho'} \gamma$ and $v_0' \simat{p_0} v_0$ we have $v \simat{p}v'$. 
\end{itemize}

For part (2), we consider the singleton and union rules:
\begin{itemize}
\item If the slicing derivation is of the form
\[
\inferrule*
{
p[\lbl], T \uneval \rho[x\mapsto p_0],e'
}
{
p, x.\{\lbl.T\} \uneval^* \rho, e', \{\lbl.p_0\}
}
\]
then recall that by assumption, $\Theta(\lbl) \supseteq \{\lbl.T\}$, so $\Theta(\lbl) = T$,
so the replay derivation must have the form 
\[\inferrule*
{
 \lbl \in \dom(\Theta) \\
 \gamma[x\mapsto v_0],T\trrun v 
}
{
 \gamma,x\in \set{\lbl.v_0}, \Theta \trrun^* \lbl \cdot v
}
\]

Now suppose that $\gamma' \simat{\rho} \gamma$ and $v_0'
\simat{\set{\lbl.p_0}} \set{\lbl.v_0}$ and $e'' \sqgeq e'$
are given where $\gamma',x\in v_0',e'' \red^* v'$.  It follows
that $v_0' = \{\lbl.v_0''\}$ and $v_0'' \simat{p_0} v_0$, so
$\gamma'[x\mapsto v_0''] \simat{\rho[x\mapsto p_0]} \gamma[x\mapsto
v_0]$.  Moreover, by inversion the derivation must have the form:
\[\inferrule*
{
 \gamma'[x\mapsto v_0''],e'' \red v'
}
{
 \gamma',x\in \set{\lbl.v_0''},e'' \red^* \lbl \cdot v'
}
\]
Since by Lemma~\ref{lem:projection-sub} $p[\lbl] \sqleq v$ we have by induction that $v' \simat{p[\lbl]} v$, and using Lemma~\ref{lem:projection-sim} we can conclude that $\lbl \cdot v'  \simat{p} \lbl \cdot v$, as desired.
\item Suppose the slicing derivation is of the form:
\[
\inferrule*
{
p|_{\dom(\Theta_1)}, x.\Theta_1 \uneval^* \rho_1,e_1',p_1\\
p|_{\dom(\Theta_2)}, x.\Theta_2 \uneval^* \rho_2,e_2',p_2
}
{
p, x.\Theta_1 \semu \Theta_2 \uneval^* \rho_1 \sqcup \rho_2,
e_0', p_1 \semu p_2
}
\]
and suppose that $\gamma, x \in v_0, \Theta \trrun^* v$ where $\Theta
\supseteq \Theta_1 \semu \Theta_2$.
Suppose that $\gamma' \simat{\rho_1 \sqcup \rho_2} \gamma$ and $v_0'
\simat{p_1 \semu p_2} v_0$ and $e_0'\sqgeq e_0$ are given, where $\gamma', x \in v_0',e_0' \red^*
v'$.  We need to show that $v' \simat{p} v$.  

Since $v_0' \simat{p_1 \semu p_2} v_0$, it is straightforward to show
that there must exist $v_1, v_2, v_1', v_2'$ such that $v_1 \semu v_2 =
v_0$, $v_1' \semu v_2' = v_0'$, $v_1 \simat{p_1} v_1$ and $v_2'
\simat{p_2} v_2$.  Furthermore, by Lemmas~\ref{lem:subset-replay} and~\ref{lem:subset-eval} we
know that $\gamma,x \in v_i, \Theta \trrun^* w_i$ and $\gamma',x \in
v_i', e_0' \red^* w_i'$ for some $w_1,w_2,w_1',w_2'$.  Therefore,
we can conclude that:

\[
\inferrule*
{
\gamma,x\in v_1,\Theta \trrun^* w_1\\
\gamma,x\in v_2,\Theta \trrun^* w_2\\
}
{
\gamma,x\in v_1\semu v_2,\Theta \trrun^* w_1\semu w_2\\
}
\qquad
\inferrule*
{
\gamma',x\in v_1',e_0' \red^* w_1'\\
\gamma',x\in v_2',e_0' \red^* w_2'\\
}
{
\gamma',x\in v_1'\semu v_2',e_0' \red^* w_1'\semu w_2'\\
}
\]
Furthermore, since $v_1 \semu v_2 = v_0$, by determinacy we know that
$w_1 \semu w_2 = v$ and similarly $w_1' \semu w_2 = v'$.  By induction
since $p_i \sqleq v_i$ and $\Theta_i \subseteq \Theta$,
we know that $w_i' \simat{p|_{\dom(\Theta_i)}} w_i$, so we know that
\[v' = w_1' \semu w_2' \simat{p|_{\dom(\Theta_i)} \semu
  p|_{\dom(\Theta_2)}} w_1 \semu w_2 = v\;.\]

To conclude, since $p \sqleq w_1 \semu w_2$ and $\dom(\Theta_1) \leq
w_1$ and $\dom(\Theta_2) \leq w_2$, it follows that $p =
p|_{\dom(\Theta_1)}\semu p|_{\dom(\Theta_2)}$, so we can conclude $w_1
\semu w_2 \simat{p} w_1' \semu w_2'$ as desired.
\end{itemize}
This exhausts all cases and completes the proof.
\end{proof}
%  LocalWords:  sublanguage toplevel judgments

%%% Local Variables: 
%%% mode: latex
%%% TeX-master: "tr"
%%% End: 

\end{document}